\documentclass[11pt]{article}
\usepackage{charter} 
\usepackage{fullpage}
\usepackage{titlesec}
\usepackage{psfrag}
\usepackage{float}
\usepackage{bbm}
\usepackage{algorithm}
\usepackage[hidelinks]{hyperref}       
\usepackage{algpseudocode}
\usepackage{url}            
\usepackage{booktabs}       
\usepackage{amsfonts}       
\usepackage{nicefrac}       
\usepackage{microtype}      
\usepackage{bm}         
\usepackage{graphicx}		
\usepackage{amsmath} 
\usepackage{cite} 
\usepackage{diagbox}
\usepackage{enumitem}
\usepackage{siunitx}
\usepackage{multirow}
\usepackage{tabularx}
\usepackage{hhline}
\usepackage{arydshln}		
\usepackage{xcolor}
\usepackage{mathtools} 
\usepackage{amsthm}			
\usepackage{amssymb}
\usepackage{setspace}

\definecolor{lightgreen}{rgb}{.9,1,.9}

\newcolumntype{L}[1]{>{\raggedright\arraybackslash}p{#1}}
\newcolumntype{C}[1]{>{\centering\arraybackslash}p{#1}}
\newcolumntype{R}[1]{>{\raggedleft\arraybackslash}p{#1}}

\theoremstyle{plain} 
\newtheorem{proposition}{Proposition}
\newtheorem{definition}{Definition}
\newtheorem{theorem}{Theorem}
\newtheorem{lemma}{Lemma}

\newtheorem{assumption}{Assumption}

\def\defn{\,\coloneqq\,}
\def\Fix{{\mathrm{Fix}}}
\def\Zer{{\mathrm{Zer}}}

\def\prox{{\mathrm{prox}}}

\def\min{\mathop{\mathrm{min}}}

\def\C{\mathbb{C}}
\def\R{\mathbb{R}}
\def\E{\mathbb{E}}

\def\zerobm{{\bm{0}}}

\def\ebm{{\bm{e}}}
\def\gbm{{\bm{g}}}

\def\xbm{{\bm{x}}}
\def\zbm{{\bm{z}}}
\def\ybm{{\bm{y}}}

\def\vbm{{\bm{v}}}
\def\zbm{{\bm{z}}}

\def\wbm{{\bm{w}}}

\def\hhat{{\widehat{h}}}
\def\phat{{\widehat{p}}}
\def\phihat{{\widehat{\phi}}}

\def\Dhat{{\widehat{D}}}
\def\Ghat{{\widehat{G}}}

\def\Abm{{\bm{A}}}
\def\Dbm{{\bm{D}}}
\def\Pbm{{\bm{P}}}
\def\Fbm{{\bm{F}}}
\def\Ibm{{\bm{I}}}

\def\thetabm{{\bm{\theta }}}
\def\Abm{{\bm{A}}}
\def\Dbm{{\bm{D}}}
\def\Pbm{{\bm{P}}}
\def\Sbm{{\bm{S}}}
\def\Fbm{{\bm{F}}}
\def\Ibm{{\bm{I}}}

\def\Dsf{{\mathsf{D}}}

\def\Isf{{\mathsf{I}}}
\def\Rsf{{\mathsf{R}}}
\def\Tsf{{\mathsf{T}}}
\def\Tsf{{\mathsf{T}}}

\def\Dsf{{\mathsf{D}}}

\def\Isf{{\mathsf{I}}}

\def\xbmast{{\bm{x}^\ast}}
\def\xbmbar{{\overline{\bm{x}}}}

\def\xbmhat{{\widehat{\bm{x}}}}

\def\argmin{\mathop{\mathrm{arg\,min}}} 

\begin{document}

\title{Deep Model-Based Architectures for Inverse Problems under Mismatched Priors}

\author{
Shirin~Shoushtari$^{1}$,~Jiaming~Liu$^{1}$, \\
Yuyang~Hu$^{1}$,~and~Ulugbek~S.~Kamilov$^{1, 2,}$\thanks{This paper is based upon work supported by the NSF CAREER award under grant CCF-2043134.}\\
\emph{\footnotesize $^{1}$Department of Electrical \& Systems Engineering,~Washington University in St.~Louis, MO 63130, USA}\\
\emph{\footnotesize $^{2}$Department of Computer Science \& Engineering,~Washington University in St.~Louis, MO 63130, USA}
}

\date{}

\maketitle

\begin{abstract}
\noindent
There is a growing interest in deep model-based architectures (DMBAs) for solving imaging inverse problems by combining physical measurement models and learned image priors specified using convolutional neural nets (CNNs). For example, well-known frameworks for systematically designing DMBAs include plug-and-play priors (PnP), deep unfolding (DU), and deep equilibrium models (DEQ). While the empirical performance and theoretical properties of DMBAs have been widely investigated, the existing work in the area has primarily focused on their performance when the desired image prior is known exactly. This work addresses the gap in the prior work by providing new theoretical and numerical insights into DMBAs under mismatched CNN priors. Mismatched priors arise naturally when there is a distribution shift between training and testing data, for example, due to test images being from a different distribution than images used for training the CNN prior. They also arise when the CNN prior used for inference is an approximation of some desired statistical estimator (MAP or MMSE). Our theoretical analysis provides explicit error bounds on the solution due to the mismatched CNN priors under a set of clearly specified assumptions. Our numerical results compare the empirical performance of DMBAs under realistic distribution shifts and approximate statistical estimators. 
\end{abstract}


\section{Introduction}

One of the most widely-studied problems in computational imaging is the recovery of an unknown image from a set of noisy measurements. The recovery problem is often formulated as an \emph{inverse problem} and solved by integrating the measurement model characterizing the response of the imaging instrument with a regularizer imposing prior knowledge on the unknown image. Some well-known image priors include nonnegativity, transform-domain sparsity, and self-similarity~\cite{Rudin.etal1992, Figueiredo.Nowak2001, Elad.Aharon2006, Danielyan2013}.

\emph{Deep Learning (DL)} has emerged in the past decade as a powerful data-driven paradigm for solving inverse problems and has improved the state-of-the-art in a number of imaging applications (see reviews~\cite{McCann.etal2017, Lucas.etal2018}). A traditional DL strategy for solving inverse problems is based on training a \emph{convolutional neural network (CNN)} to perform a regularized inversion of the forward model, thus leading to a mapping from the measurements to the unknown image. For example, U-Net~\cite{Ronneberger.etal2015} and DnCNN~\cite{Zhang.etal2017} are two prototypical architectures used for designing traditional DL methods for solving imaging inverse problems. 

There is a growing interest in \emph{deep model-based architectures (DMBAs)} for inverse problems that integrate physical measurement models and CNN image priors (see reviews~\cite{Ongie.etal2020, Ahmad.etal2020, Monga.etal2021, Kamilov.etal2022}). Well-known DMBAs that explicitly account for the measurement models include \emph{plug-and-play priors (PnP)}, \emph{deep unfolding (DU)}, \emph{compressive sensing using generative models (CSGM)}, and \emph{deep equilibrium architectures (DEQ)}~\cite{Venkatakrishnan.etal2013, Sreehari.etal2016, Gregor.LeCun2010, Bora.etal2017, Gilton.etal2021}. DMBAs are systematically obtained from model-based iterative algorithms by parametrizing the regularization step as a CNN and training it to adapt to the empirical distribution of desired images. An important conceptual point about typical DMBAs is that they do \emph{not} solve an optimization problem. That is, even when the original model-based algorithm solves an optimization problem, once the regularizer is replaced with a CNN, then there is no longer any corresponding function to minimize. Remarkably, the heuristic of using CNNs not associated with any explicit regularization function exhibited great empirical success and spurred much theoretical and algorithmic work~\cite{Ongie.etal2020, Ahmad.etal2020, Monga.etal2021, Kamilov.etal2022}. 

Despite the rich literature on DMBAs, the existing work in the area has primarily focused on settings where the desired image prior is known \emph{exactly}. While this assumption has led to many useful algorithms and insights, it fails to capture the range of situations arising in imaging inverse problems. Specifically, the knowledge of the image prior is only approximate if there is a distribution shift between training and testing data, for example, due to testing images being from a different distribution than images used for training the CNN prior. Alternatively, the CNN prior used for inference within DMBA might be an approximation of some desired true statistical estimator, such as \emph{maximum a posteriori probability (MAP)} estimator or \emph{minimum mean squared error (MMSE)} estimator. In both of these settings, it would be valuable to gain insights on how the discrepancies in CNN priors influence the discrepancies in estimated images.

In this work, we address this gap by providing a set of new theoretical and numerical results into DMBAs under mismatched CNN priors. We focus on the architecture derived from the \emph{steepest descent} variant of \emph{regularization by denoising (SD-RED)}~\cite{Romano.etal2017} by considering two types of CNN priors: (a) image denoisers trained to remove \emph{additive white Gaussian noise (AWGN)}; (b) \emph{artifact removal (AR)} operators learned end-to-end using DEQ. Our theoretical analysis provides explicit error bounds on the solution obtained by SD-RED due to the mismatched CNN priors under a set of explicitly specified assumptions. Our numerical results illustrate the practical influence of mismatched CNN priors on image recovery from subsampled Fourier measurements, which is a well-known problem in accelerated \emph{magnetic resonance imaging (MRI)}~\cite{Lustig.etal2008}. Specifically, we provide numerical results on two related but distinct scenarios where: (i) CNN priors are trained on data mismatched to the testing data and (ii) CNN priors are trained to approximate an explicit image regularizer.


\section{Background}

\noindent
\textbf{Inverse problems.} We consider the problem of recovering an unknown image $\xbmast \in \R^n$ from its  measurements $\ybm \in \R^m$. The problem is traditionally formulated as an inverse problem where the solution is computed by solving an optimization problem 
\begin{equation}
\label{Eq:OptimizationForInverseProblem}
\xbmhat = \argmin_{\xbm \in \R^n} f(\xbm) \quad\text{with}\quad f(\xbm) = g(\xbm) + h(\xbm),
\end{equation}
where $g$ is the \emph{data-fidelity term} enforcing consistency of the solution with $\ybm$ and $h$ is the \emph{regularizer} enforcing prior knowledge on $\xbm$. The formulation in eq.~\eqref{Eq:OptimizationForInverseProblem} corresponds to the MAP estimator when
\begin{equation}
\label{Eq:MAP}
g(\xbm) = - \log (p_{\ybm|\xbm}(\xbm)) \quad \text{and} \quad h(\xbm) = - \log(p_{\xbm}(\xbm))
\end{equation}
where $p_{\ybm|\xbm}$ is the likelihood relating $\xbm$ to measurements $\ybm$ and $p_{\xbm}$ is the prior distribution. For example, given measurements of the form $\ybm = \Abm\xbm+\ebm$, where $\Abm$ is the \emph{measurement operator} (also known as the \emph{forward operator}) characterizing the response of the imaging instrument and $\ebm$ is AWGN, the data-fidelity term reduces to the quadratic function $g(\xbm) = \frac{1}{2}\|\ybm-\Abm\xbm\|_2^2$. On the other hand, a widely-used sparsity promoting regularizer in imaging inverse problems is \emph{total variation (TV)} $h(\xbm) = \tau \|\Dbm\xbm\|_1$, where $\Dbm$ is the image gradient and $\tau > 0 $ controls the strength of regularization~\cite{Rudin.etal1992, Bioucas-Dias.Figueiredo2007, Beck.Teboulle2009}. 

\noindent
\textbf{Model-based optimization.} Proximal algorithms are often used for solving problems of form~\eqref{Eq:OptimizationForInverseProblem} when $h$ is nonsmooth (see the review~\cite{Parikh.Boyd2014}). Two widely used families of proximal algorithms for imaging inverse problems are the \emph{proximal gradient method (PGM)}~\cite{Figueiredo.Nowak2003 , Daubechies.etal2004, Bect.etal2004, Beck.Teboulle2009} and the \emph{alternating direction method of multipliers (ADMM)}~\cite{Eckstein.Bertsekas1992, Afonso.etal2010, Ng.etal2010, Boyd.etal2011}. Both PGM and ADMM avoid differentiating $h$ by using the \emph{proximal operator}, which can be defined as
\begin{equation}
\label{eq:proximaloperator}
\prox_{\sigma^2 h} (\zbm) := \argmin_{\xbm \in \R^n} \bigg\{ \frac{1}{2} \|\xbm -\zbm\|^2_2 + \sigma^2 h(\xbm) \bigg\}, \quad \sigma > 0,
\end{equation}
for any proper, closed, and convex function $h$~\cite{Parikh.Boyd2014}. Comparing eq.~\eqref{eq:proximaloperator} and eq.~\eqref{Eq:OptimizationForInverseProblem}, we see that the proximal operator can be interpreted as a MAP estimator for the AWGN denoising problem
\begin{equation}
\label{Eq:MAPDenoise}
\zbm = \xbm_0 + \wbm \quad\text{where}\quad \xbm_0 \sim p_{\xbm_0}\,, \quad \wbm \sim \mathcal{N}(0, \sigma^2 \Ibm)\,,
\end{equation}
by setting $h(\xbm) = -\log(p_{\xbm_0}(\xbm))$. It is worth noting that another less known but equally valid statistical interpretation of the proximal operator is as a MMSE estimator~\cite{Gribonval.etal2012, Gribonval.Machart2013}.

\medskip\noindent
\textbf{PnP and RED.} PnP~\cite{Venkatakrishnan.etal2013, Sreehari.etal2016} and RED~\cite{Romano.etal2017} are two related families of iterative algorithms that enable integrating measurement operators and CNN priors for solving imaging inverse problems (see the recent review of PnP in~\cite{Kamilov.etal2022}). Since for general denoisers PnP/RED do not solve an optimization problem~\cite{Reehorst.Schniter2019}, it is common to interpret PnP/RED as fixed-point iterations of some high-dimensional operators. For example, given a denoiser $\Dsf_\thetabm: \R^n \rightarrow \R^n$ parameterized by a CNN with weights $\thetabm$, the iterations of SD-RED~\cite{Romano.etal2017} can be written as
\begin{equation}
\label{Eq:REDItetation}
\xbm^k = T_\thetabm(\xbm^{k-1}) = \xbm^{k-1} - \gamma G_\thetabm(\xbm^{k-1}) \quad\text{with}\quad G_\thetabm(\xbm) \defn \nabla g(\xbm) + \tau(\xbm - D_\thetabm(\xbm))\;,
\end{equation}
where $g$ is the data-fidelity term, and  $\gamma, \tau > 0$ are the step size and the regularization parameters, respectively. SD-RED seeks to compute a fixed-point $\xbmbar \in \R^n$ of the operator $T_\thetabm$, which is equivalent to finding a zero of the operator $G_\thetabm$
\begin{equation}
\label{Eq:ZeroRED}
\xbmbar\in\Fix(T_\thetabm)\defn\{\xbm\in\R^n:T_\thetabm(\xbm)=\xbm\}\quad\Leftrightarrow\quad G_\thetabm(\xbmbar) = \nabla g(\xbmbar) + \tau(\xbmbar-D_\thetabm(\xbmbar))=\zerobm\;,
\end{equation}
The solutions of~\eqref{Eq:ZeroRED} balance the requirements to be both data-consistent (via $\nabla g$) and noise-free (via $(I-D_\thetabm)$), which can be intuitively interpreted as finding an equilibrium between the physical measurement model and CNN prior. Remarkably, this heuristic of using denoisers not necessarily associated with any $h$ within an iterative algorithm exhibited great empirical success~\cite{Zhang.etal2017a, Metzler.etal2018, Dong.etal2019, Zhang.etal2019, Mataev.etal2019, Sun.etal2019b, Liu.etal2020, Ahmad.etal2020, Wei.etal2020, Xie.etal2021} and spurred a great deal of theoretical work on PnP/RED~\cite{Chan.etal2016, Meinhardt.etal2017, Buzzard.etal2017, Reehorst.Schniter2019, Ryu.etal2019, Sun.etal2018a, Tirer.Giryes2019, Teodoro.etal2019, Xu.etal2020, Tang.Davies2020, Sun.etal2021, Cohen.etal2020}. Recent line of PnP work has also explored the parameterization of the regularization functions as CNNs, thus leading to explicit loss functions~\cite{Cohen.etal2021, Hurault.etal2022}.

\medskip\noindent
\textbf{DU and DEQ.} DU (also known as \emph{deep unrolling} or \emph{algorithm unrolling}) is a DMBA paradigm highly related to PnP/RED with roots in sparse coding \cite{Gregor.LeCun2010}. DU has gained popularity in inverse problems due to its ability to provide a systematic connection between iterative algorithms and deep neural network architectures~\cite{Ongie.etal2020, Monga.etal2021}. The SD-RED algortihm~\eqref{Eq:REDItetation} can be turned into a DU architectures by truncating it to a fixed number of iterations $t \geq 1$, and training the corresponding architecture end-to-end in a supervised fashion by comparing the predicted image $\xbm^t(\thetabm)$ to the ground-truth $\xbmast$. 
DEQ is a recent extension of DU that can accommodate an arbitrary number of iterations~\cite{bai.etal2019, Gilton.etal2021}. DEQ can be implemented for SD-RED by replacing $\xbm^t(\thetabm)$ in the DU loss by a fixed-point $\overline{\xbm}(\thetabm)$ in eq.~\eqref{Eq:ZeroRED} and using implicit differentiation to update the weights $\thetabm$. The benefit of DEQ over DU is that it doesn't require the storage of the intermediate variables in training, thus reducing the memory complexity. However, DEQ requires the computation of the fixed-point $\overline{\xbm}(\thetabm)$, which can increase the computational complexity.

\medskip\noindent
\textbf{Other related work.} There were a number of recent publications exploring the topic of distribution shifts in inverse problems. The use of DMBAs for adapting pre-trained CNNs to shifts in the measurement model has been discussed in several publications~\cite{Liu.etal2020, Xie.etal2021, Gilton.etal2021a}. The performance gap due to distribution shifts on several well-known DL architectures has been empirically quantified for accelerated MRI in~\cite{Darestani.etal2021}. \emph{Test-time training} was proposed as a strategy to decrease the performance gap for certain distribution shifts in~\cite{Darestani.etal2022}. The robustness of \emph{compressive sensing (CS)} recovery using mismatched distributions with bounded Wasserstein distances was analyzed in~\cite{Jalal.etal2021}. The robustness of CSGM to changes in the ground-truth distribution and measurement operator in CS-MRI was investigated in~\cite{Jalal.etal2021a}. Finally, it is worth to briefly note a distinct line of work in optimization exploring the impact of \emph{inexact} proximal operators on the convergence of the traditional proximal algorithms~\cite{Bertsekas2011, Schmidt.etal2011, Devolder.etal2013}.

\medskip\noindent
\textbf{Our contributions.} Despite their conceptual differences in training, PnP, DU, and DEQ can be implemented using the same architecture during inference. For example, the SD-RED iteration in eq.~\eqref{Eq:REDItetation} can be interpreted as a PnP method when the CNN prior $D_\thetabm$ is an AWGN denoiser, as a DU architecture when $D_\thetabm$ was trained using a fixed number of unfolded iterations, and as a DEQ architecture when $D_\thetabm$ was trained at a fixed point. Note that the image prior in DU and DEQ is not an AWGN denoiser, but rather an artifact removal (AR) operator $D_\thetabm$ trained by taking into account the distribution of artifacts within the iterations of SD-RED. The existing work on PnP, DU, and DEQ has primarily focused on the settings where the inference is performed assuming that $D_\thetabm$ exactly corresponds to the desired AWGN denoiser or AR operator. However, it is clear from the discussion above, that this is an idealized scenario, in particular, when there is a distribution shift between training and testing data or when $D_\thetabm$ is an approximation of some true statistical estimator. Our contribution is a first technical discussion on this issue in the scenario where the CNN prior used for inference in SD-RED is an approximation of some true prior. While we use the SD-RED iterations as the basis for our DMBA and corresponding mathematical analysis, the results can be extended to other DMBAs, including those based on PnP-PGM or PnP-ADMM. In short, this work presents new theoretical analysis and numerical results that are both complementary to and backward compatible with the existing literature in the area.


\section{Theoretical Analysis}

We focus on the MBDA based on the following modified SD-RED iteration
\begin{equation}
\label{Eq:ApproxSDRED}
\xbm^k = \xbm^{k-1} - \gamma \widehat{G}(\xbm^{k-1}) \quad\text{with}\quad \Ghat(\xbm) \defn \nabla g(\xbm) + \tau (\xbm - \Dhat(\xbm)),
\end{equation}
where we refer to $\Dhat$ as a \emph{mismatched prior} that approximates some \emph{desired} or \emph{true prior} $D$. We denote as $\xbmast \in \Zer(G)$ the solution of SD-RED in~\eqref{Eq:REDItetation} using the true $D$. We write both operators as $D_\sigma$ and $\Dhat_\sigma$ when explicitly highlighting the \emph{strength parameter} $\sigma > 0$ used to control the regularization strength. This parameter can account for the variance $\sigma^2$ in the proximal operator in eq.~\eqref{eq:proximaloperator} and is analogous to the standard deviation parameter in the traditional PnP methods~\cite{Venkatakrishnan.etal2013}. We next present a theoretical analysis of SD-RED under a mismatched prior providing: (a) error bounds on the solutions computed by~\eqref{Eq:ApproxSDRED} and (b) statistical interpretations under $\Dhat$ approximating a proximal operator, possibly corresponding to MAP or MMSE estimators.
Our theoretical analysis builds on a set of explicitly specified assumptions that serve as sufficient conditions.
\begin{assumption}
\label{As:ConvDataFit}
The function $g$ is convex, continuous, and has a Lipschitz continuous gradient with constant $L > 0$.
\end{assumption}
This is a standard assumption in optimization and is relatively mild in the context of imaging inverse problems. For example, it is satisfied by many traditional data-fidelity terms, including those based on the least-squares loss.
\begin{assumption}
\label{As:IdealDenCont}
The operator $D$ is Lipschitz continuous with constant $0 < \lambda \leq 1$.
\end{assumption}
The Lipschitz continuity of CNN priors has been extensively considered in the prior work on PnP, DU, and DEQ and can be practically implemented using spectral normalization methods (see~\cite{Kamilov.etal2022} for a more detailed discussion in the context of PnP). We say that $D$ is a contractive operator when $\lambda < 1$ and it is a nonexpansive operator when $\lambda = 1$. Note also how the nonexpansiveness is \emph{only} assumed on the desired CNN prior $D$ rather than the mismatched one $\Dhat$ used for inference.
\begin{assumption}
\label{As:BoundedError}
The operator $\widehat{D}_\sigma$ satisfies
$$\|\widehat{D}_\sigma(\xbm)-D_\sigma(\xbm)\|_2 \leq \sigma \varepsilon, \quad \text{for all }\xbm \in \R^n,$$
where $\sigma > 0$ is the strength parameter of the prior and $\varepsilon > 0$ is some constant.
\end{assumption}
This assumption bounds the distance between the true and mismatched priors. We explicitly relate the bound to $\sigma$, since for small values of the strength parameter $\sigma$ it is natural for the CNN priors to act as identity. The constant $\varepsilon$ quantifies the approximation ability of the mismatched prior; given two approximations, the one with smaller $\varepsilon$ is expected to be a better match. Assumption~\ref{As:BoundedError} can be also justified by using statistical considerations. For example, Theorem~\ref{Thm:DistShift} below shows that when $D$ and $\Dhat$ are MAP denoisers, Assumption~\ref{As:BoundedError} is a direct consequence of a bound on the density ratio $p_\xbm/\phat_\xbm$. A natural consequence of this argument is that when both $D$ and $\Dhat$ are available, Assumption~\ref{As:BoundedError} can be a proxy to quantify prior distribution shifts.

We are now ready to state the first result.
\begin{theorem}
\label{Thm:Contraction}
Run SD-RED in~\eqref{Eq:ApproxSDRED} for $t \geq 1$ iterations under Assumptions~\ref{As:ConvDataFit}-\ref{As:BoundedError} with $\lambda < 1$ using a fixed step-size
$$0 < \gamma < \frac{(1-\lambda)\tau}{(L+(1+\lambda)\tau)^2}.$$
Then, there exists a unique $\xbmast \in \Zer(G)$ such that
$$\|\xbm^t - \xbmast\|_2 \leq \eta^t R_0 + \tau \sigma \varepsilon A,$$
where $0 < \eta < 1$, $R_0 \defn \|\xbm^0-\xbmast\|_2$, and $A \defn \gamma/(1-\eta)$ are constants independent of $t$.
\end{theorem}
The proof is provided in Appendix~\ref{Sec:Thm:Contraction}. Theorem~\ref{Thm:Contraction} shows that SD-RED using a mismatched CNN prior---either an AWGN denoiser or an AR operator---can approximate $\xbmast \in \Zer(G)$ up to an error term proportional to $\tau, \sigma$, and $\varepsilon$. Note that it is expected that the error shrinks for small values of $\tau$ and $\sigma$ since they control the influence of the CNN prior. While Theorem~\ref{Thm:Contraction} assumes that the true prior $\Dsf_\sigma$ is a contraction, our next result relaxes this condition by adopting an additional assumption.
\begin{assumption}
\label{As:BoundedIter}
The operator $G$ is such that $\Zer(G) \neq \varnothing$. There exists a finite number $R > 0$ such that for any $\xbmast \in \Zer(G)$, the iterates~\eqref{Eq:ApproxSDRED} satisfy
$$\|\xbm^t-\xbmast\|_2 \leq R\quad\text{for all}\quad t \geq 1.$$
\end{assumption}
\noindent
Assumption~\ref{As:BoundedIter} simply states that $G$ has a zero point in $\R^n$, which is equivalent to the assumption that SD-RED using the true prior has a solution. The assumption additionally states that the iterates generated via eq.~\eqref{Eq:ApproxSDRED} are bounded, which is natural for many imaging problems, since images usually have bounded pixel values in $[0, 1]$ or $[0, 255]$.
\begin{theorem}
\label{Thm:NonExp}
Run SD-RED in~\eqref{Eq:ApproxSDRED} for $t \geq 1$ iterations under Assumptions~\ref{As:ConvDataFit}-\ref{As:BoundedIter} with $\lambda = 1$ using a fixed step-size
$0 < \gamma < 1/(L+2\tau)$.
Then, we have that
$$\min_{1 \leq i \leq t} \|G(\xbm^{i-1})\|_2^2 \leq \frac{1}{t}\sum_{i = 1}^t \|G(\xbm^{i-1})\|_2^2 \leq \frac{B_1}{t} + \tau\sigma \varepsilon B_2,$$
where $B_1 \defn ((L+2\tau)R^2)/\gamma$ and $B_2 \defn (L+2\tau)(2 R + \gamma \tau\sigma \varepsilon)$ are constants independent of $t$.
\end{theorem}
The proof is provided in Appendix~\ref{Sec:Thm:NonExp}. Theorem~\ref{Thm:NonExp} is more general since it relaxes the assumption that $D$ is a contraction to it being a nonexpansive operator. This comes with a cost of a slower sublinear convergence of the first term, compared to the linear convergence of the corresponding term under a contractive prior. Both theorems are compatible with the prior work on DMBAs in the sense that by setting $\varepsilon = 0$ one recovers the traditional convergence results in the literature~\cite{Sun.etal2018a, Sun.etal2019b, Ryu.etal2019, Gilton.etal2021}.

Theorems~\ref{Thm:Contraction} and~\ref{Thm:NonExp} establish general error bounds on the solutions of SD-RED using a mismatched operator $\Dhat$, under the assumption that for the same input the distance between the outputs of $\Dhat$ and $D$ are bounded. The following result provides a statistical interpretation to this assumption by considering denoisers that perform MAP estimation.
\begin{theorem}
\label{Thm:DistShift}
Let $p_\xbm$ and $\phat_\xbm$ denote two log-concave continuous probability density functions supported over $\R^n$, and $D_\sigma$ and $\Dhat_\sigma$ denote corresponding MAP estimators for the AWGN problem in eq.~\eqref{Eq:MAPDenoise}. Let $r \defn p_\xbm/\phat_\xbm$ denote the density ratio of $p$ and $\phat$. If $\mathrm{exp}(-\varepsilon^2/2) \leq r(\xbm) \leq \mathrm{exp}(\varepsilon^2/2)$ for all $\xbm \in \R^n$, then we have that  $\|D_\sigma(\xbm)-\Dhat_\sigma(\xbm)\|_2 \leq \sigma \varepsilon$ for all $\xbm \in \R^n$.
\end{theorem}
The proof is provided in Appendix~\ref{Sec:Thm:DistShift}. Theorem~\ref{Thm:DistShift} shows that if two prior distributions $p_\xbm$ and $\phat_\xbm$ are close to each other, then the distance between corresponding MAP denoisers is small, finally leading to a small error terms in Theorems~\ref{Thm:Contraction} and~\ref{Thm:NonExp}. It is worth mentioning that the density ratio is a common tool for quantifying the distances between probability densities and is used, for example, in the Kullback–Leibler divergence  $\E_{p_\xbm}[\log(r(\xbm))]$.

The next result enables a statistical interpretation to Theorems~\ref{Thm:Contraction} and~\ref{Thm:NonExp}. This is due to the fact that both MAP and MMSE denoisers for the AWGN problem in eq.~\eqref{Eq:MAPDenoise} can be expressed as proximal operators~\cite{Gribonval2011, Gribonval.Machart2013, Xu.etal2020}. Using this result one can obtain a novel interpretation of PnP, DU, and DEQ algorithms under mismatched priors as using CNNs approximating true priors corresponding to some statistical estimators. For the rest of this section, we set $\tau = 1/\sigma^2$ to simplify the mathematical analysis and consider the true prior of form $D_\sigma = \prox_{\sigma^2 h}$, where the function $h$ satisfies the following assumption.
\begin{assumption}
\label{As:LipschitzReg}
The function $h$ is closed, proper, convex, and Lipschitz continuous with constant $S > 0$.
\end{assumption}
This assumption is commonly adopted in nonsmooth optimization and implies the existence of a global upper bound on subgradients~\cite{ouyang2013, Yu2013, Boyd.Vandenberghe2008}. It is satisfied by a large number of functions, including the $\ell_1$-norm and TV regularizers. We are now ready to state the final result, which can be seen as an extension of the analysis in~\cite{Sun2019b} to mismatched CNN priors.
\begin{theorem}
\label{Thm:ApproxStatEst}
Run SD-RED in~\eqref{Eq:ApproxSDRED} for $t \geq 1$ iterations under Assumptions~\ref{As:ConvDataFit}-\ref{As:LipschitzReg} using a fixed step-size
$0 < \gamma < 1/(L+2\tau)$.
Then, we have that
$$\min_{1 \leq i \leq t} (f(\xbm^{i-1})-f^\ast) = \frac{1}{t}\sum_{i = 1}^t (f(\xbm^{i-1}) - f^\ast) \leq \frac{2(L+2\tau)R^3}{\gamma t} + \frac{\epsilon^2 R}{\sigma^2} + \frac{S^2 \sigma^2}{2} ,$$
where $f = g + h$, $D_\sigma = \prox_{\sigma^2 h}$, and $\tau = 1/\sigma^2$.
\end{theorem}
The proof is given in Appendix~\ref{Sec:Thm:ApproxStatEst}. It states that SD-RED using a mismatched CNN prior $\Dhat_\sigma$, which approximates some proximal operator (possibly corresponding to MAP or MMSE statistical estimators), can approximate the minimum $f^\ast$ up to an error term that is a function of $\epsilon$ and $\sigma^2$. The error term is minimized to $\epsilon S \sqrt{2R}$ when $\sigma^2 = \sqrt{2\epsilon^2 R/S^2}$.

To conclude this section, we theoretically analyzed the SD-RED iteration in eq.~\eqref{Eq:ApproxSDRED}, where a mismatched prior $\Dhat$ is used instead of the desired or true prior $D$. Our analysis shows that one can get explicit error bounds on the solution of DMBAs that depend on the level of mismatch. Our analysis also reveals  in the context of MAP estimation that the obtained error bounds can be related to the prior distribution shifts. In the next section, we provide a set of numerical results illustrating the empirical performance of SD-RED under distribution shifts and approximate proximal operators.


\section{Numerical Evaluation}

There has been significant interest in understanding the performance of CNN priors for image recovery from noisy and sub-sampled measurements. The recent work on DMBAs has shown that CNN priors can lead to significant improvements over traditional image priors such as TV in a wide variety of inverse problems. The results presented in this section evaluate the performance under mismatched CNN priors obtained due to shifts in the data distribution (for example, training on natural images but testing on MRI images) and due to approximate proximal operators (for example, training a CNN as an empirical MAP or MMSE estimator). It is worth mentioning that our focus is to highlight the impact of mismatched CNN priors, not to justify SD-RED as a DL architecture (such justifications can be found elsewhere, for example, see~\cite{Romano.etal2017, Reehorst.Schniter2019, Kamilov.etal2022}).

We present three distinct sets of simulations on image recovery from subsampled Fourier measurements: \emph{(a) using severely mismatched CNN priors corresponding to AWGN denoisers and AR operators trained on MRI, CT, and natural images}; \emph{(b) using CNNs trained to approximate the traditional TV proximal operator}; and \emph{(c) using moderately mismatched CNN priors for MRI trained on a different anatomical regions or MRI modalities}. While our simulations focus on subsampled measurements without noise, we expect that the reported relative performances will be preserved for moderate levels of noise. In all simulations, we use the classical TV regularization as a representative of the traditional model-based image reconstruction~\cite{Beck.Teboulle2009}. 

\begin{figure}[t]
\centering\includegraphics[width=0.5\textwidth]{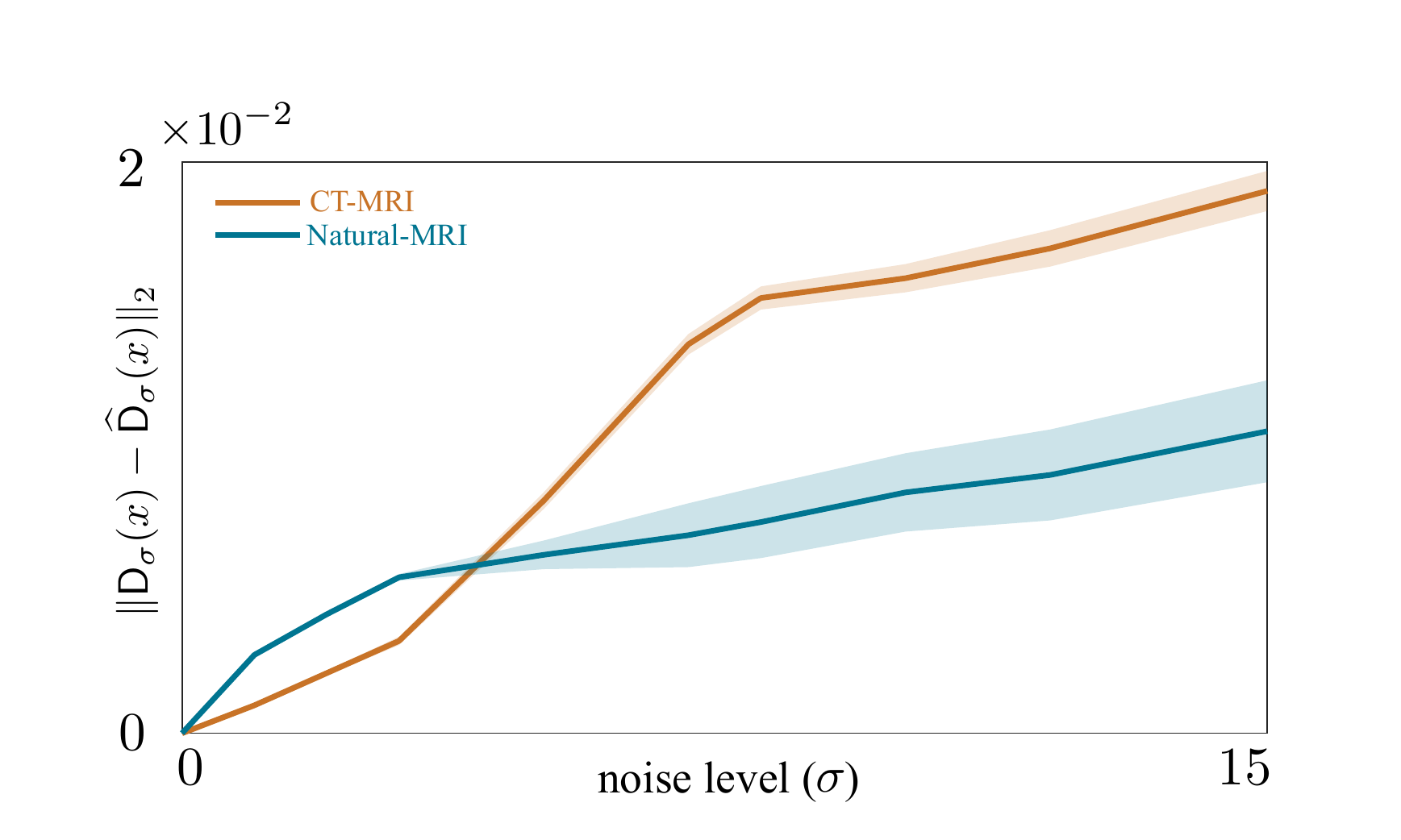}
\caption{\small Numerical evaluation of the mismatch between CT and natural image CNN priors with respect to the true MRI prior on MRI test images at various noise levels $\sigma$. As supposed in Assumption~\ref{As:BoundedError} the distance between the mismatched and true CNN priors is bounded and proportional to the noise level $\sigma$.}
\label{Fig:Dist}
\end{figure}

\subsection{Mismatched Priors from Training on CT, Natural, and MR Images}

We consider image recovery from subsampled Fourier measurements $\ybm = \Abm \xbm \in \C^m$, where $\Abm$ performs radial Fourier sampling~\cite{Lustig.etal2008}. The measurement operator can be written as $\Abm = \Pbm \Fbm$, where $\Fbm$ is the Fourier transform and $\Pbm$ as a diagonal sampling matrix. We train three classes of CNN priors modeling different empirical data distrbutions: (a) natural grayscale images from~\cite{Martin.etal2001}, (b) brain images from~\cite{zhang2018ista}, and (c) CT images from the \emph{low dose CT grand challenge} of Mayo Clinic~\cite{mccollough2016tu}.
Ten  $180 \times 180$ images from Set11 were randomly chosen as natural test images. From 50 slices of $256 \times 256$ images  provided for testing in~\cite{zhang2018ista}, ten random images were chosen as MRI test images. Ten random CT images of $512 \times 512$ were chosen as CT test images. The recovery performance of SD-RED using all three classes of CNN priors, as well as the traditional TV regularizer~\cite{Beck.Teboulle2009}, is quantified using \emph{peak-signal-to-noise ratio (PSNR)} in dB and \emph{structural similarity index measure (SSIM)}. 

Beyond the data distribution considerations, we also consider two types of CNN priors: (i) AWGN denoisers and (ii) AR operators. AWGN denoisers are extensively used within PnP due to their simplicity and effectiveness, while AR operators have been widely reported to achieve state-of-the-art imaging performance~\cite{Kamilov.etal2022, Liu.etal2021b}. We train one AWGN denoiser for each noise level $\sigma \in \{1, 2, 3, 5, 7, 8, 10, 12, 15\}$ and image distribution (natural, MRI, and CT) using the DnCNN architecture with batch normalization layers removed (as was done in~\cite{Liu.etal2021b}).  The architecture of the AR operators is identical to that of the AWGN denoisers. The AR operators are trained vis DEQ using pre-trained AWGN denoisers for initialization~\cite{Gilton.etal2021}. We train one AR operator for each sampling ratio considered in the simulations ($10\%$, $20\%$, and $30\%$) and image distribution (natural, MRI, and CT). Nesterov~\cite{Nesterov2004} and Anderson~\cite{Anderson.etal2009} acceleration techniques are used in forward and backward DEQ iterations for faster convergence. Spectral normalization is used for controlling the Lipschitz constants of all our CNNs~\cite{Miyato.etal2018, Liu.etal2021b}. For each experiment, we select $\sigma$ and $\tau$ achieving the best PSNR performance.

\begin{figure}[t]
\centering\includegraphics[width=0.9\textwidth]{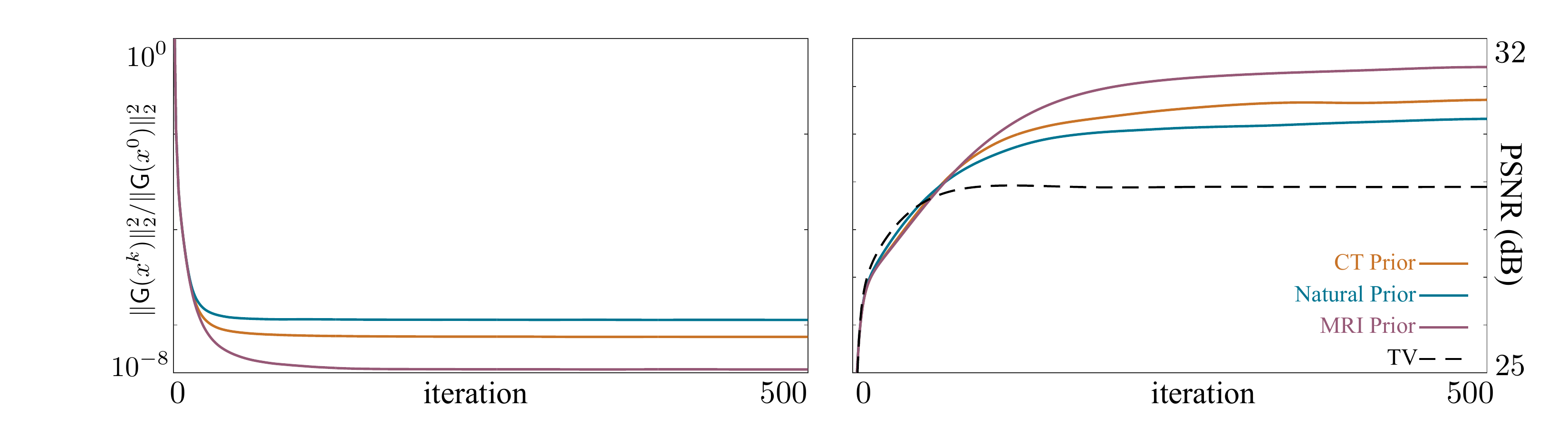}
\caption{\small \textbf{Left:} Empirical evaluation of convergence of SD-RED using the true and mismatched AWGN priors for brain MRI reconstruction at $10\%$ sampling. Average normalized distance to $\Zer(G)$ is plotted against the iterations. \textbf{Right:} PSNR (dB) is plotted against iterations for TV and three AWGN priors. Note the effect of mismatched priors on the convergence of SD-RED to $\Zer(G)$ and the superior performance of mismatched CNN priors over TV.}
\label{Fig:gnorm}
\vspace{-.5em}   
\end{figure}

\begin{figure}[t]
\centering\includegraphics[width=1.01\textwidth]{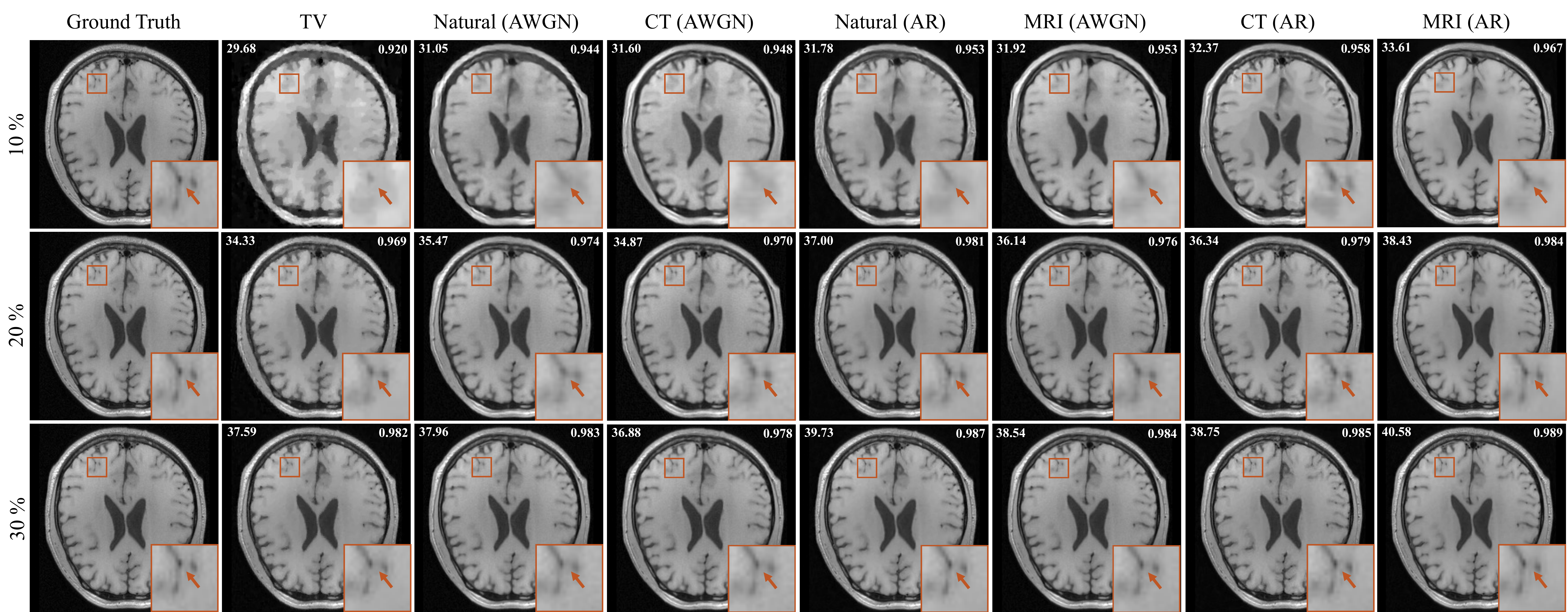}
\caption{\small Visual evaluation of several AWGN and AR priors used within SD-RED for brain MRI reconstruction from radially subsampled Fourier measurements at $10\%$ , $20\%$, and $30\%$. Note how all the learned priors (both true and mismatched) outperform the traditional TV prior. Additionally, note the performance improvement due to the mismatched AR priors compared to the true and mismatched AWGN priors.}
\label{Fig:mri}
\vspace{-.5em}   
\end{figure}

Figure~\ref{Fig:Dist} plots the distances between the outputs of the true AWGN denoisers trained on MRI images and mismatched ones trained on CT and natural images. The plot is generated using MRI test images. Average distance of the denoisers is plotted against noise level $\sigma$. Shaded area illustrates the range of values obtained across all test images. Our theoretical analysis assumes that the distance between the outputs of CNN priors is bounded and proportional to $\sigma$. Figure~\ref{Fig:Dist} numerically shows that the distance between the CNN priors is indeed bounded with a constant proportional to the noise level.

\begin{figure}[t]
\centering\includegraphics[width=1\textwidth]{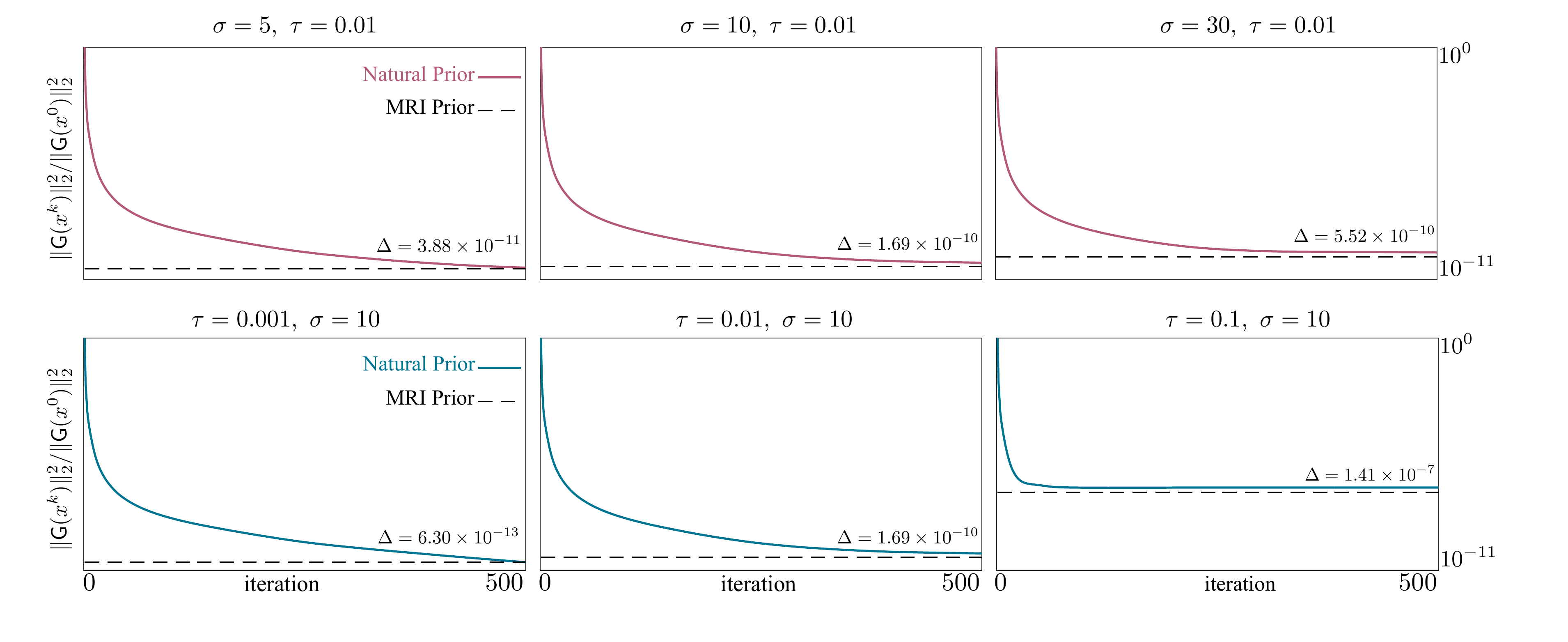}
\caption{\small Illustration of influence of noise level ($\sigma$) and regularization parameter ($\tau$) on the error term in the convergence analysis of SD-RED under mismatched priors. Average normalized distance to $\Zer(G)$ is plotted against iterations for MRI and natural AWGN priors for MRI reconstruction at $10\%$ sampling. \textbf{Top:} $\tau$ is set to a constant to evaluate the effect of $\sigma$. \textbf{Bottom:} $\sigma$ is set to a constant to evaluate the effect of $\tau$. The gap illustrates the error term due to the use of mismatched CNN priors. Note how the gap increases with the increase of both parameters $\tau$ and $\sigma$.}
\label{Fig:tau_sigma}
\vspace{-.5em}   
\end{figure}

Theorem \ref{Thm:NonExp} states that SD-RED with shifted priors converges to an element of $\Zer (G)$ up to an error term that depends on $\tau$, $\sigma$, and $\epsilon$. Figure~\ref{Fig:gnorm} illustrates the convergence of SD-RED with CT, natural, and MRI AWGN priors on MRI test data under $10\%$ subsampling ratio. The average value of 
$\|G(\xbm^k)\|_2^2/\|G(\xbm^0)\|_2^2$ and PSNR (dB) are plotted against iterations of the algorithm. The distance to the zero-set quantifies the convergence behavior of the algorithm and is expected to be smaller for matched priors compared to mismatched priors. For reference, we also provide the evolution of the PSNR values obtained using TV reconstruction. Figure~\ref{Fig:tau_sigma} shows the influence of $\tau$ and $\sigma$ on the convergence of SD-RED using a mismatched natural AWGN prior to $\Zer(G)$, where $G$ uses the true MRI AWGN prior at $10\%$ subsampling. The average value of 
$\|G(\xbm^k)\|_2^2/\|G(\xbm^0)\|_2^2$  is plotted against iteration number for $\sigma \in \{5, 10, 30\}$ and $\tau \in \{0.001 , 0.01 , 0.1\}$. The results are consistent with the theoretical analysis and show that increase in both $\tau$ and $\sigma$ increases the error term.

Table~\ref{Tab:table1} reports the recovery PSNR for natural, CT, and MRI test images from subsampled Fourier measurements using the corresponding true and mismatched CNN priors. The best PSNR values for AWGN and AR priors are highlighted separately in bold. Figure~\ref{Fig:mri} presents corresponding visual comparisons on a MRI image with $10\%$, $20\%$, and $30\%$ sampling. As expected, matched priors lead to better performance in all the experiments, with matched AR priors achieving the best performance. While mismatched priors result in performance drops for both AWGN and AR priors, we observe better overall results for AR priors. Moreover, it can be observed that the recovery at higher sampling ratios is less vulnerable to the distribution shifts in the prior. Finally, note the inferior performance of TV compared to SD-RED under both true and mismatched CNN priors. 

\begin{table}[t]
\caption{The recovery of natural, MRI, and CT images in terms of PSNR (dB) using the true and mismatched AR and AWGN priors.}
    \centering
    \renewcommand\arraystretch{1.2}
    {\footnotesize
    \scalebox{0.96}{
    \begin{tabular*}{13.45cm}{L{68pt} L{26pt}lL{26pt}lC{26pt}lC{26pt}lC{26pt}l}
        \hline
        \multirow{2}{4em}{\textbf{Method}}& \multicolumn{3}{c}{\textbf{MRI} }& \multicolumn{3}{c}{\textbf{Natural}} & \multicolumn{3}{c}{\textbf{CT}} \\
        \cline{2-10}
        &\multicolumn{1}{c}{10\%}  & \multicolumn{1}{c}{20\%} &\multicolumn{1}{c}{30\%} & \multicolumn{1}{c}{10\%} &\multicolumn{1}{c}{20\%} &\multicolumn{1}{c}{30\%} &\multicolumn{1}{c}{10\%} &\multicolumn{1}{c}{20\%} &\multicolumn{1}{c}{30\%} &\\\cline{1-10}
        
        \textbf{TV}       & {28.89}  &   {33.64} & \multicolumn{1}{c:}{36.79}  
        & {23.12}  & {27.35} & \multicolumn{1}{c:}{30.47}
        &{31.58}     &{39.11} &{41.39}\\
        \cdashline{1-10}
         \textbf{MRI (AWGN)} & {\textbf{31.40}} & {\textbf{35.53}} & \multicolumn{1}{c:}{\textbf{37.92}}
         & {24.15}&{29.26} & \multicolumn{1}{c:}{\textbf{32.01}}
         &{35.39} & {39.86} & {42.49} \\  
         \textbf{Natural (AWGN)} & {30.32} & {34.80}&\multicolumn{1}{c:}{37.36}
         & {\textbf{25.15}}&{\textbf{29.35}} & \multicolumn{1}{c:}{31.99}
         &{35.41} & {39.76} & {42.32} \\ 
         \textbf{CT (AWGN)} & {30.71} & {34.23} & \multicolumn{1}{c:}{36.27}
         & {22.45}&{28.39} & \multicolumn{1}{c:}{31.71}
         &\textbf{35.68} & \textbf{39.97} & \textbf{42.74} \\ 
         \cdashline{1-10}
         \textbf{MRI (AR)} & {\textbf{32.39}} & {\textbf{37.45}} & \multicolumn{1}{c:}{\textbf{39.73}}
         & {22.73}&{27.59} & \multicolumn{1}{c:}{31.27}
         &{35.24} & {\textbf{42.61}} & {46.00} \\ 
         \textbf{Natural (AR)} & {31.08} & {36.11} & \multicolumn{1}{c:}{38.80}
         & {\textbf{25.28}}&{\textbf{29.92}} & \multicolumn{1}{c:}{\textbf{33.03}}
         &{36.35} & {42.32} & {45.89} \\ 
         \textbf{CT (AR)} & {31.11} & {35.44} & \multicolumn{1}{c:}{37.96}
         & {21.39}&{27.17} & \multicolumn{1}{c:}{28.79}
         &{\textbf{38.35}} & {42.47} & {\textbf{46.94}} \\ \hline
        
    \end{tabular*}}
    }
\label{Tab:table1}
\end{table}

\subsection{Approximating the TV proximal operator using a CNN}

In this section, we numerically evaluate CNN priors trained to approximate proximal operators. To that end, we train DnCNN to approximate the TV proximal operator and use the trained network within SD-RED as a mismatched CNN prior. We will refer to the mismatched prior as \emph{TV Approx} and the true prior as \emph{TV Exact}. We generate the training dataset by adding AWGN to natural grayscale images from~\cite{Martin.etal2001}. We pre-train one CNN for each noise level $\sigma \in \{1,2,3,5,7, 8, 10, 12, 15\}$ by using the TV solution with optimized regularization parameter as the training label. We consider the same recovery problem as in the previous section, where the goal is to recover an image from its subsampled Fourier measurements at $10\%$, $20\%$, and $30\%$ sampling rates. Theorem \ref{Thm:ApproxStatEst} shows that SD-RED using \emph{TV Approx} can approximate the solution of \eqref{Eq:OptimizationForInverseProblem} up to an error term. This behaviour is illustrated in Figure~\ref{Fig:objective} (Left) for natural images reconstructed at $10\%$ subsampling.  The average value of $ f(\xbm ^k) /f(\xbm ^0)$ is plotted against iterations. Figure~\ref{Fig:objective} (Right) shows that the approximate TV prior can indeed achieve performance similar to the true TV prior. These plots highlight that despite the constant error term in the objective function due to the prior mismatch, the approximation can still lead to nearly identical PSNR and SSIM values. Figure~\ref{Fig:inexact} illustrates the recovery of two test images at three sampling rates, highlighting the ability of TV Approx to match the performance of the true TV prior.

\begin{figure}[t]
\centering\includegraphics[width=1.0\textwidth]{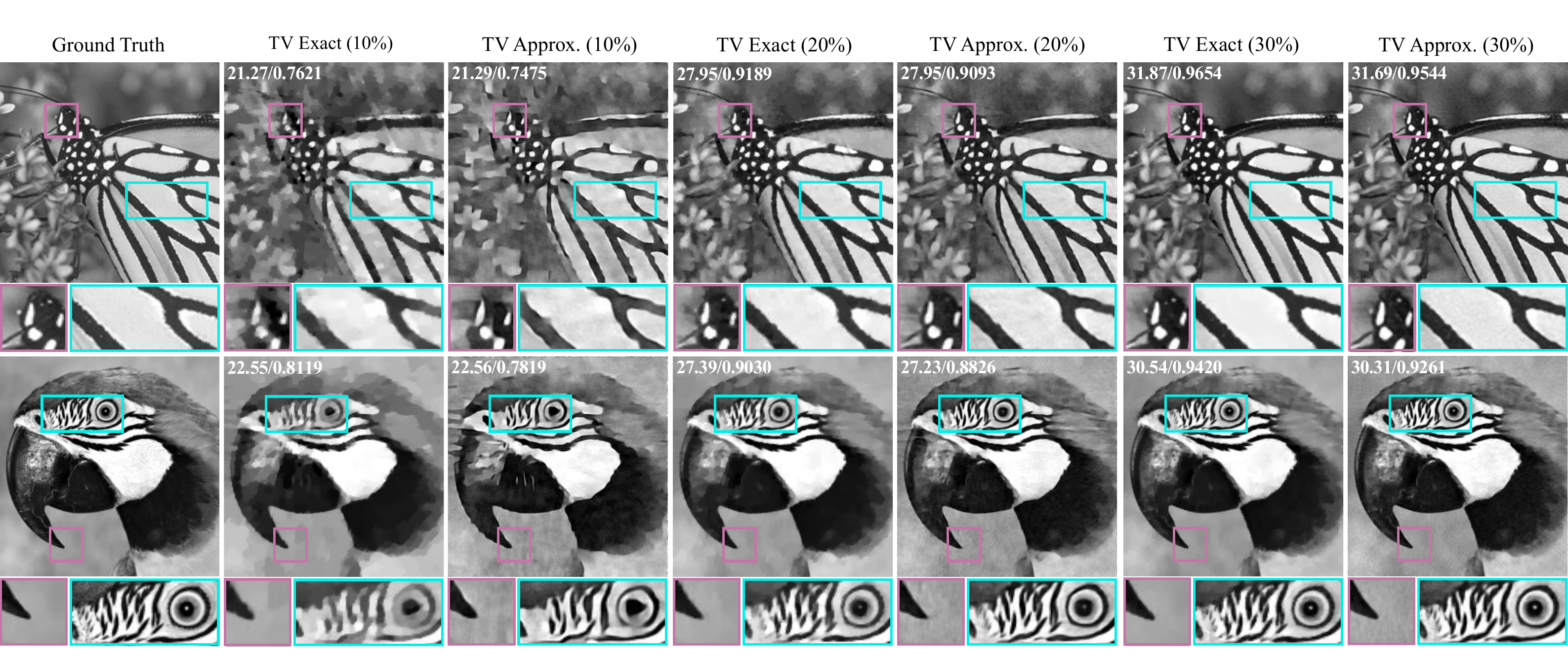}
\caption{\small Recovery of \emph{Butterfly} and \emph{Parrot} from $10\%$, $20\%$, and $30\%$ Fourier samples using SD-RED under DnCNN trained as an approximate TV prior. Results of the traditional TV are also provided. Note the visual and quantitative similarities between the exact and approximate TV results at all sampling ratios.}
\label{Fig:inexact}
\vspace{-.5em}   
\end{figure}
 \begin{figure}[t]
\centering\includegraphics[width=0.8\textwidth]{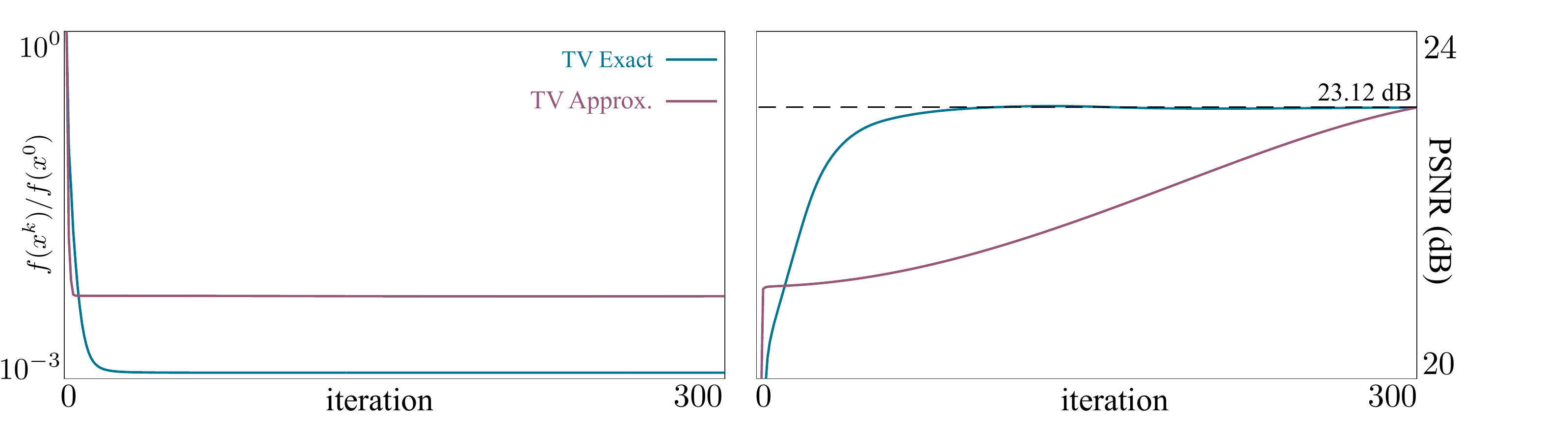}
\caption{\small \textbf{Left:}~Evolution of the average TV objective using the true TV prior and its DnCNN approximation over the test set of natural images. \textbf{Right:} Evolution of average PSNR values on the true and approximate TV priors over the same test set of natural images. Note that despite the constant error term predicted in the objective due to the use of mismatched priors, the approximate TV prior achieves nearly identical PSNR compared to the true TV prior.}
\label{Fig:objective}
\vspace{-.5em}   
\end{figure}

\subsection{Using Mismatched MRI Priors for Accelerated Parallel MRI}

In this section, we consider only MRI images to evaluate the impact of CNN priors trained on data with moderate distribution shifts. Our measurement model uses multi-coil Cartesian Fourier sampling mask from the fastMRI challenge \cite{knoll2020fastmri} with $2\times$ and $4\times$ accelerated acquisitions. 
The measurement operator for each coil can be written as $\Abm_i = \Pbm \Fbm \Sbm_i $, where $\Pbm$ is a subsampling mask, $\Fbm$ is the Fourier transform, and $\Sbm_i$ is a coil sensitivity map. We use datasets in~\cite{zhang2018ista} and~\cite{knoll2020fastmri} to pre-train four separate AWGN denoisers on brain, knee, AXT1PRE, and AXT2 images. The denoisers are trained on noise levels corresponding to $\sigma \in \{1, 2, 3, 5, 8, 10, 12, 15\}$. We select the denoiser yielding the best performance in terms of PSNR as the prior for SD-RED. Ten images are randomly chosen from each dataset for testing. In each experiment, 128 coil sensitivity maps are synthesized using SigPy \cite{Ong.etal2019}.
Table~\ref{Tab:table2} reports PSNR and SSIM for $2\times$ and $4\times$ accelerated MRI reconstruction using SD-RED. Figure~\ref{Fig:APMRI} presents a visual comparison on one test image from the fastMRI AXT2 dataset using several CNN priors. The results indicate that mismatched CNN priors can be useful in the settings where true CNN priors are not available. In particular, one can observe how mismatched priors can outperform the traditional TV reconstruction in all settings.

\begin{figure}[t]
\centering\includegraphics[width=1\textwidth]{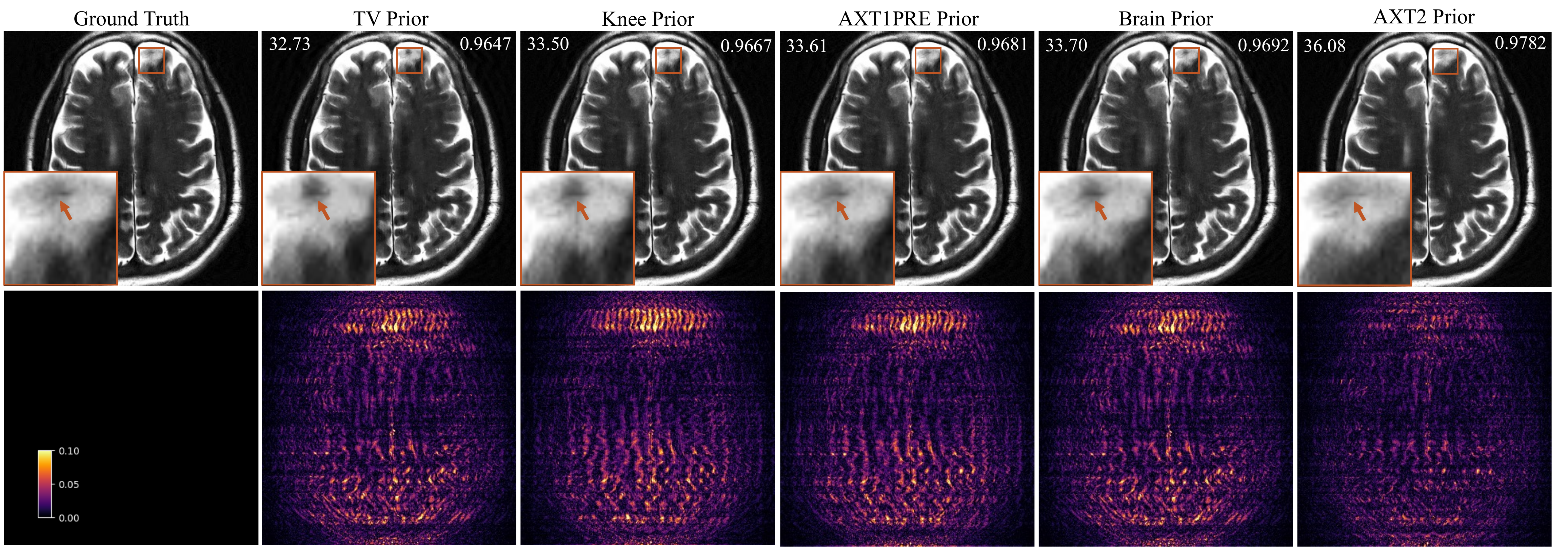}
\caption{\small Image recovery from $4\times$ accelerated parallel MRI measurements of an image from AXT2 fastMRI dataset using several CNN priors trained as AWGN denoisers. Note how despite the performance drop due to the distribution mismatched, all CNN priors significantly outperform the traditional TV regularizer.}
\label{Fig:APMRI}
\vspace{-.5em}   
\end{figure}

\begin{table}[t]
\caption{Recovery of MRI AXT2 and AXT1PRE test images from accelerated parallel MRI meak using several image priors.}
    \centering
    \renewcommand\arraystretch{1.2}
    {\footnotesize
    \scalebox{0.96}{
    \begin{tabular*}{13.2cm}{L{70pt} C{31pt}lC{31pt}lC{31pt}lC{31pt}lC{31pt}lC{31pt}lC{31pt}l}
        \hline
        \multirow{2}{4em}{\textbf{Method}}& \multicolumn{4}{c}{\textbf{MRI AXT2~\cite{knoll2020fastmri}} } & \multicolumn{4}{c}{\textbf{MRI AXT1PRE~\cite{knoll2020fastmri}}} \\
        \cline{2-9}
        &\multicolumn{2}{c}{2x}  & \multicolumn{2}{c}{4x} & \multicolumn{2}{c}{2x} & \multicolumn{2}{c}{4x} \\\hline
        \textbf{TV} & {40.78} & \multicolumn{1}{c:}{0.993}& {33.30}  & \multicolumn{1}{c:}{0.968}& {45.06} & \multicolumn{1}{c:}{0.995} & {36.70}&{0.968} \\
         \cdashline{1-9}
                
         \textbf{Knee Prior}       & 41.99          & \multicolumn{1}{c:}{0.994}           &   34.37        & \multicolumn{1}{c:}{0.971}     & 46.70 & \multicolumn{1}{c:}{0.996}           & 38.56  & 0.981 \\
         \textbf{AXT2 Prior}      &  \textbf{44.09}         & \multicolumn{1}{c:}{\textbf{0.996}}        &    \textbf{36.10} & \multicolumn{1}{c:}{\textbf{0.977} }    
         &  47.03& \multicolumn{1}{c:}{0.977}           & 38.84 &  0.982 \\ 
          \textbf{AXT1PRE Prior}    & \textcolor{black}{41.77}            & \multicolumn{1}{c:}{\textcolor{black}{0.994}}            & {34.36}            & \multicolumn{1}{c:}{0.970}   & \textbf{47.84} & \multicolumn{1}{c:}{\textbf{0.997}}           & \textbf{39.65}   &  \textbf{0.984}\\
         \textbf{Brain MRI Prior}      & {41.61}            & \multicolumn{1}{c:}{0.993}            & {34.03}            & \multicolumn{1}{c:}{0.967}   & 46.97 & \multicolumn{1}{c:}{0.977}            & 38.83   & 0.981
            \\\hline
    \end{tabular*}}
    }
\label{Tab:table2}
\end{table}


\section{Conclusion}

In this work, we explored the topic of mismatched CNN priors within DMBAs by theoretically analyzing the error bounds due to the mismatch and numerically illustrating the impact of the mismatch on the imaging performance. While our focus has been on the SD-RED architecture, similar results can certainly be carried out using other DMBAs. Our results show how the severity of the mismatch on the CNN prior translates to that of the final recovered images, relate the mismatch in CNN priors to the distribution shifts in statistical priors, and highlight the potential of DMBAs using mismatched CNN priors to outperform the traditional TV regularizer. In the future work, we would like to explore alternative characterizations of the mismatch and develop methods to reduce the influence of mismatch on the final recovery performance.


\section{Appendix}

We adopt \emph{monotone operator theory}~\cite{Bauschke.Combettes2017, Ryu.Boyd2016} as a framework for a unified analysis of SD-RED under mismatched priors. In Appendix~\ref{Sec:Thm:Contraction}, we analyze SD-RED under the assumption that the true CNN prior $D$ is a contraction. In Appendix~\ref{Sec:Thm:NonExp}, we extend the analysis to nonexpansive operators $D$. In Appendix~\ref{Sec:Thm:DistShift}, we show how the bound on the density ratio $r = p_\xbm/\phat_\xbm$ leads to the bound in Assumption~\ref{As:BoundedError}. In Appendix~\ref{Sec:Thm:ApproxStatEst}, we analyze SD-RED under approximate proximal operators. In Appendices~\ref{Sec:UsefulResults} and~\ref{Sec:BackgroundTheory}, we present some key results from monotone operator theory and traditional optimization that are useful for our analysis. 


\subsection{Proof of Theorem~\ref{Thm:Contraction}}
\label{Sec:Thm:Contraction}

First note that
$$\|G(\xbm)-\Ghat(\xbm)\|_2 = \tau\|(I-D)(\xbm)-(I-\Dhat) (\xbm)\|_2 \leq \tau \sigma \epsilon.$$
Consider a single iteration $\xbm^+ = \xbm-\gamma \Ghat(\xbm)$ and $\xbmast \in \Zer(G)$
\begin{align*}
\|\xbm^+-\xbmast\|_2 &= \|\xbm-\gamma \Ghat(\xbm)-\xbmast\|_2 \leq \|\xbm-\gamma G(\xbm)-\xbmast\|_2 + \gamma \|G(\xbm)-\Ghat(\xbm)\|_2 \\
&\leq \eta \|\xbm-\xbmast\|_2 + \gamma \tau \sigma \epsilon,
\end{align*}
where in the first inequality we used the triangular inequality and in the second Proposition~\ref{Prop:MonotoneConvIter} in Appendix~\ref{Sec:UsefulResults}.
By iterating this inequality for $t \geq 1$ iterations, we obtain
$$\|\xbm^t-\xbmast\|_2 \leq \eta^t\|\xbm^0-\xbmast\|_2 + \tau \sigma \epsilon A,$$
where $\eta^2 \defn 1-2\gamma [\tau(1-\lambda)] + \gamma^2 [L+(1+\lambda)\tau]^2 \in (0, 1)$ and $A \defn \gamma/(1-\eta)$ are fixed constants.


\subsection{Proof of Theorem~\ref{Thm:NonExp}}
\label{Sec:Thm:NonExp}

Consider a single iteration $\xbm^+ = \xbm-\gamma \Ghat(\xbm)$ and $\xbmast \in \Zer(G)$ 
\begin{align*}
\|\xbm^+-\xbmast\|_2^2 
&= \|\xbm - \gamma \Ghat(\xbm)-\xbmast\|_2^2 \\
&= \|\xbm - \gamma G(\xbm)-\xbmast\|_2^2 + 2\gamma (G(\xbm)-\Ghat(\xbm))^\Tsf(\xbm - \gamma G(\xbm)-\xbmast) + \gamma^2 \|G(\xbm)-\Ghat(\xbm)\|_2^2 \\
&\leq \|\xbm-\xbmast\|_2^2 - \frac{\gamma}{L+2\tau}\|G(\xbm)\|_2^2 + 2\gamma \|G(\xbm)-\Ghat(\xbm)\|_2 \|\xbm - \gamma G(\xbm)-\xbmast\|_2 + \gamma^2 \|G(\xbm)-\Ghat(\xbm)\|_2^2 \\
&\leq \|\xbm-\xbmast\|_2^2 - \frac{\gamma}{L+2\tau}\|G(\xbm)\|_2^2 + 2\gamma \tau \sigma \epsilon  R + \gamma^2 \tau^2 \sigma^2 \epsilon^2.
\end{align*}
By rearranging the terms, we get
$$\|G(\xbm)\|_2^2 \leq \left(\frac{L+2\tau}{\gamma}\right)\left[\|\xbm-\xbmast\|_2^2 - \|\xbm^+-\xbmast\|_2^2\right] + (L+2\tau)(2\tau\sigma \epsilon R+\gamma\tau^2\sigma^2\epsilon^2).$$
By averaging over $t \geq 1$ iterations, we get
$$\frac{1}{t}\sum_{i = 1}^t \|G(\xbm^{i-1})\|_2^2 \leq \frac{B_1}{t} + \tau \sigma \epsilon B_2,$$
where $B_1 \defn ((L+2\tau)R^2)/\gamma$ and $B_2 \defn (L+2\tau)(2 R + \gamma \tau\sigma \epsilon)$.


\subsection{Proof of Theorem~\ref{Thm:DistShift}}
\label{Sec:Thm:DistShift}

The MAP solution to~\eqref{Eq:MAPDenoise} can be expressed as the proximal operator~\eqref{eq:proximaloperator}. Consider two density functions $p_\xbm$ and $\phat_\xbm$ and corresponding MAP regularizers $h(\xbm) = -\log(p_\xbm(\xbm))$ and $\hhat(\xbm) = -\log(\phat_\xbm(\xbm))$. The log-concavity and continuity of $p_\xbm$ and $\phat_\xbm$ imply that $\prox_{\sigma^2 h}$ and $\prox_{\sigma^2 \hhat}$ are unique minimizers of $1$-strongly convex functions $\phi(\xbm) = \frac{1}{2}\|\xbm-\zbm\|_2^2 + \sigma^2 h(\xbm)$ and $\phihat(\xbm) = \frac{1}{2}\|\xbm-\zbm\|_2^2 + \sigma^2 \hhat(\xbm)$, respectively. From the definition of the proximal operator, $\phi$ and $\phihat$ are minimized at $\xbmast = D_\sigma(\zbm) = \prox_{\sigma^2 h}(\zbm)$ and $\xbmhat = \Dhat_\sigma(\zbm) = \prox_{\sigma^2 \hhat}(\zbm)$, respectively, where $\zbm \in \R^n$ is any vector. From strong convexity
\begin{equation}
\label{Eq:VecFuncBound}
\begin{matrix}
\phi(\xbmhat) \geq \phi(\xbmast) + \frac{1}{2}\|\xbmast-\xbmhat\|_2^2 \\
\phihat(\xbmast) \geq \phihat(\xbmhat) + \frac{1}{2}\|\xbmast-\xbmhat\|_2^2
\end{matrix}
\quad\Rightarrow\quad \|\xbmast-\xbmhat\|_2^2 \leq \sigma^2 (\hhat(\xbmast)-h(\xbmast)+h(\xbmhat) - \hhat(\xbmhat)).
\end{equation}
We can re-write the bound on the density ratio as
\begin{align*}
&\mathrm{e}^{-\varepsilon^2/2} \leq p_\xbm(\xbm)/\phat_\xbm(\xbm) \leq  \mathrm{e}^{\varepsilon^2/2} \quad\Rightarrow\quad -\varepsilon^2/2 \leq \log(p_\xbm(\xbm)) - \log(\phat_\xbm(\xbm)) \leq \varepsilon^2/2 \quad\Rightarrow\quad |h(\xbm)-\hhat(\xbm)| \leq \varepsilon^2/2.
\end{align*}
By combining this inequality with~\eqref{Eq:VecFuncBound}, we get the desired conclusion
$$\|D_\sigma(\zbm) - \Dhat_\sigma(\zbm)\|_2^2 \leq \sigma^2 \epsilon^2,$$
which is true for any $\zbm \in \R^n$.


\subsection{Proof of Theorem~\ref{Thm:ApproxStatEst}}
\label{Sec:Thm:ApproxStatEst}

Moreau smoothing is well-known in the optimization literature and has been extensively used to design and analyze non-smooth algorithms (see, for example, \cite{Yu2013}). It has been previously used in Theorem~2 of~\cite{Sun2019b} to analyze the block-coordinate RED (called BC-RED). The contribution of the following analysis is to extend~\cite{Sun2019b} the prior result to a mismatched CNN prior $\Dhat$. Following~\cite{Sun2019b}, we fix $\tau = 1/\sigma^2$ for convenience, which enables us to have only one regularization parameter, namely $\sigma^2$.

We define the following loss function to approximate $f = g + h$
$$f_{\sigma^2}(\xbm) = g(\xbm) + \tau h_{\sigma^2}(\xbm),$$
where we set $\tau = (1/\sigma^2)$ and $h_{\sigma^2}$ is known as the Moreau envelope of $h$
\begin{equation}
\label{Eq:MorSmoothBound}
h_{\sigma^2}(\xbm) \defn \min_{\vbm \in \R^n} \left\{\frac{1}{2}\|\vbm-\xbm\|_2^2 + \sigma^2 h(\vbm)\right\}.
\end{equation}
Lemma~\ref{Lem:UniformBoundMoreau} and Assumption~\ref{As:LipschitzReg} imply that
\begin{equation}
\label{Eq:Thm4SmoothingBound}
0 \leq h(\xbm) - \tau h_{\sigma^2}(\xbm) \leq \frac{S \sigma^2}{2} \quad\Leftrightarrow\quad 0 \leq f(\xbm) - f_{\sigma^2}(\xbm) \leq \frac{S \sigma^2}{2}.
\end{equation}
Additionally, Lemma~\ref{Lem:GradMorProxRes} implies that
\begin{equation}
\label{Eq:SDREDProxUpd}
G_\sigma(\xbm) = \nabla g(\xbm) + \tau (\xbm-\prox_{\sigma^2 h}(\xbm)) = \nabla f_{\sigma^2}(\xbm),
\end{equation}
where $\nabla f_\sigma^2$ is $(L+2\tau)$-Lipschitz continuous. Eq.~\eqref{Eq:SDREDProxUpd} implies that a single iteration of SD-RED using the true CNN prior is a gradient step on the smoothed version $f_{\sigma^2}$ of $f$. From~\eqref{Eq:MorSmoothBound} and the convexity of the Moreau envelope, we have for all $\xbm \in \R^n$
\begin{equation}
\label{Eq:Thm4MinBound}
f_{\sigma^2}^\ast = f_{\sigma^2}(\xbmast) \leq f_{\sigma^2}(\xbm) \leq f(\xbm),
\end{equation}
where $\xbmast \in \Zer(G)$. Hence, there exists a finite $f^\ast$ such that $f(\xbm) \geq f^\ast$ with $f^\ast \geq f_{\sigma^2}^\ast$ for all $\xbm \in \R^n$.

Consider a single SD-RED update using a mismatched prior
$\xbm^+ = \xbm - \gamma \Ghat(\xbm)$
\begin{align}
\nonumber f_{\sigma^2}(\xbm^+) &\leq f_{\sigma^2}(\xbm) + \nabla f_{\sigma^2}^\Tsf(\xbm^+ - \xbm) + \frac{(L+2\tau)}{2}\|\xbm^+ - \xbm\|_2^2 \\
\nonumber&= f_{\sigma^2}(\xbm) - \gamma \nabla f_{\sigma^2}(\xbm)^\Tsf\Ghat(\xbm) + \frac{\gamma^2(L+2\tau)}{2}\|\Ghat(\xbm)\|_2^2 \\
\nonumber&\leq f_{\sigma^2}(\xbm) + \frac{\gamma}{2}\left[\|\Ghat(\xbm)\|_2^2 - 2 \nabla f_{\sigma^2}(\xbm)^\Tsf\Ghat(\xbm)\right]\\
\label{Eq:Thm4GradBound}&\leq f_{\sigma^2}(\xbm) - \frac{\gamma}{2}\|\nabla f_{\sigma^2}(\xbm)\|_2^2 + \frac{\gamma \tau^2\sigma^2\epsilon^2}{2},
\end{align}
where in first inequality we used the Lipschitz continuity of $\nabla f_{\sigma^2}$, in the second the fact that $\gamma \leq 1/(L+2\tau)$, and the third
$$\|\Ghat(\xbm)-\nabla f_{\sigma^2}(\xbm)\|_2 \leq \tau\sigma\epsilon \quad\Leftrightarrow\quad \|\Ghat(\xbm)\|_2^2 - 2\nabla f_{\sigma^2}(\xbm)^\Tsf\Ghat(\xbm) \leq (\tau \sigma \epsilon)^2 - \|\nabla f_{\sigma^2}(\xbm)\|_2^2.$$
Consider the iteration $t \geq 1$, then we have that
\begin{align*}
&\min_{1 \leq i \leq t} (f(\xbm^{i-1})-f^\ast) \leq \frac{1}{t} \sum_{i = 1}^t (f(\xbm^{i-1})-f^\ast) \leq \frac{1}{t} \sum_{i = 1}^t (f_{\sigma^2}(\xbm^{i-1})-f_{\sigma^2}^\ast) + \frac{S^2\sigma^2}{2} \\
&\leq  \frac{R}{t} \sum_{i = 1}^t \|\nabla f_{\sigma^2}(\xbm^{i-1})\|_2^2 + \frac{S^2 \sigma^2}{2} \leq \frac{2R}{\gamma t} (f_{\sigma^2}(\xbm^0)-f_{\sigma^2}^\ast) + \frac{\epsilon^2R}{\sigma^2} + \frac{S^2 \sigma^2}{2} \\
&\leq \frac{2 R^3 (L+2\tau)}{\gamma t} + \frac{\epsilon^2 R}{\sigma^2} + \frac{S^2 \sigma^2}{2},
\end{align*}
where in the second inequality we used eq.~\eqref{Eq:Thm4SmoothingBound} and~\eqref{Eq:Thm4MinBound}, in the third the convexity of $f_{\sigma^2}$ and Assumption~\ref{As:BoundedIter}, in the forth eq.~\eqref{Eq:Thm4GradBound} and that $\tau = 1/\sigma^2$, and in the final the convexity of $f_{\sigma^2}$ and Assumption~\ref{As:BoundedIter}.


\subsection{Useful results for the main theorems}
\label{Sec:UsefulResults}

\begin{proposition}
\label{Prop:CoCoerciveG}
Suppose Assumptions~\ref{As:ConvDataFit}-\ref{As:IdealDenCont} are true. Then, $G = \nabla g + \tau R$  is $(1/(L+2\tau))$-cocoercive.
\end{proposition}

\begin{proof}
Since $\nabla g$ is $L$-Lipschitz continuous, from Lemma~\ref{Lem:CoCoerciveGrad} is also $(1/L)$-cocoercive. Then, from Lemma~\ref{Lem:MonotoneRelationships}, $\Isf-(2/L)\nabla g$ is nonexpansive.

\medskip\noindent
Since $D = I - R$ is nonexpansive and any convex combination of nonexpansive operators is nonexpansive, we have that the following operator is also nonexpansive
\begin{align*}
I - \frac{2}{L+2\tau}G
&= \left[\frac{2}{L+2\tau} \cdot \frac{2\tau}{2}\right](I-R) + \left[\frac{2}{L+2\tau}\cdot \frac{L}{2}\right] (I - (2/L)\nabla g)\\
&= \left[1 - \frac{2}{L+2\tau}\cdot \frac{L}{2}\right](I-R) + \left[\frac{2}{L+2\tau}\cdot \frac{L}{2}\right] (I - (2/L)\nabla g).
\end{align*}
Thus, Lemma~\ref{Lem:MonotoneRelationships} implies that $G$ is $1/(L+2\tau)$-cocoercive.
\end{proof}

\begin{proposition}
Suppose Assumptions~\ref{As:ConvDataFit}-\ref{As:IdealDenCont} are true. Then, for any $0 < \gamma \leq 1/(L+2\tau)$, we have
$$\|\xbm-\gamma G(\xbm) - \xbmast\|_2^2 \leq \|\xbm-\xbmast\|_2^2 - \left(\frac{\gamma}{L+2\tau}\right)\|G(\xbm)\|_2^2,$$
where $\xbm \in \R^n$ and $\xbmast \in \Zer(G)$.
\end{proposition}

\begin{proof}
We have the following set of relations
\begin{align*}
\|\xbm-\gamma G(\xbm)-\xbmast\|_2^2 
&= \|\xbm-\xbmast\|_2^2 - 2\gamma G(\xbm)^\Tsf(\xbm-\xbmast) + \gamma^2\|G(\xbm)\|_2^2 \\
&\leq \|\xbm-\xbmast\|_2^2 - \left(\frac{2\gamma}{L+2\tau}\right) \|G(\xbm)\|_2^2 + \left(\frac{\gamma}{L+2\tau}\right)\|G(\xbm)\|_2^2 \\
&=\|\xbm-\xbmast\|_2^2 - \left(\frac{\gamma}{L+2\tau}\right)\|G(\xbm)\|_2^2,
\end{align*}
where in the second row we used Proposition~\ref{Prop:CoCoerciveG} and the fact that $\gamma \leq 1/(L+2\tau)$.
\end{proof}

\begin{proposition}
\label{Prop:MonotoneG}
Suppose Assumptions~\ref{As:ConvDataFit}-\ref{As:IdealDenCont} are true with $\lambda < 1$. Then, $G = \nabla g + \tau R$  is $[\tau(1-\lambda)]$-strongly monotone and $[L+(1+\lambda)\tau]$-Lipschitz continuous.
\end{proposition}

\begin{proof}
From the convexity and smoothness of $g$, we know that $\nabla g$ is monotone. Due to the contractiveness of $D$, Lemma~\ref{Lem:StrongMonotoneResidual} implies that $R$ is $(1-\lambda)$-strongly monotone. Thus, we have that
\begin{align*}
(G(\xbm)-G(\zbm))^\Tsf(\xbm-\zbm) 
&= (\nabla g(\xbm)-\nabla g(\zbm))^\Tsf(\xbm-\zbm) + \tau (R(\xbm)-R(\zbm))^\Tsf(\xbm-\zbm) \\
&\geq \tau (1-\lambda) \|\xbm-\zbm\|_2^2.
\end{align*}
To see the Lipschitz continuity, note that
$$\|G(\xbm)-G(\ybm)\|_2 \leq \|\nabla g(\xbm)-\nabla g(\ybm)\|_2 + \tau \|R(\xbm)-R(\ybm)\|_2 \leq L + \tau (1+\lambda).$$
\end{proof}

\begin{proposition}
\label{Prop:MonotoneConvIter}
Suppose Assumptions~\ref{As:ConvDataFit}-\ref{As:IdealDenCont} are true with $\lambda < 1$. Suppose that the step-size parameter is selected to satisfy
$$0 < \gamma < \frac{(1-\lambda)\tau}{(L+(1+\lambda)\tau)^2}.$$
Then, for any $\xbm \in \R^n$ and $\xbmast \in \Zer(G)$, we have
$$\|\xbm-\gamma G(\xbm)-\xbmast\|_2^2 \leq \eta^2 \|\xbm-\xbmast\|_2^2,$$
where $\eta^2 = 1-2\gamma [\tau(1-\lambda)] + \gamma^2 [L+(1+\lambda)\tau]^2 \in (0, 1)$.
\end{proposition}

\begin{proof} 
Let $\ell = L + (1+\lambda)\tau$ and $\mu = (1-\lambda)\tau$.
\begin{align*}
\|\xbm-\gamma G(\xbm)-\xbmast\|_2^2 
&= \|\xbm-\xbmast\|_2^2 - 2\gamma G(\xbm)^\Tsf(\xbm-\xbmast) + \gamma^2 \|G(\xbm)\|_2^2\\
&\leq \|\xbm-\xbmast\|_2^2 - 2\gamma \mu \|\xbm-\xbmast\|_2^2 + \gamma^2 \ell^2 \|\xbm-\xbmast\|_2^2 \\
&= \eta^2 \|\xbm-\xbmast\|_2^2,
\end{align*}
where $\eta^2 = (1-2\gamma\mu + \gamma^2\ell^2)$. Thus, for any $0 < \gamma < 2\mu/\ell^2$, we have that $0 < \eta^2 < 1$.

\end{proof}

\begin{proposition}
\label{Prop:MinimizationUsingProx}
Suppose Assumptions~\ref{As:ConvDataFit}-\ref{As:LipschitzReg}

the update
$$\xbm^+ = \xbm - \gamma \Ghat(\xbm), \quad \xbm \in \R^n,$$
under , where $D_\sigma = \prox_{\sigma^2 h}$ and $0 < \gamma \leq 1/(L+2\tau)$.
Then, for any , we have
$$\|\nabla f_{(1/\tau)}(\xbm)\|_2^2 \leq \frac{2}{\gamma} (f_{(1/\tau)}(\xbm)-f_{(1/\tau)}(\xbm^+)) + \sigma^2\varepsilon^2,$$
\end{proposition}


\subsection{Background material}
\label{Sec:BackgroundTheory}

The results in this section are well-known in the optimization literature and can be found in different forms in standard textbooks~\cite{Nesterov2004, Bauschke.Combettes2017, Rockafellar1970, Boyd.Vandenberghe2004}. We summarize the results useful for our analysis by restating them in a convenient form.

\begin{lemma}
\label{Lem:CoCoerciveGrad}
For convex and continuously differentiable function $g$, we have
$$\text{$\nabla g$ is $L$-Lipschitz continuous} \quad\Leftrightarrow\quad \text{$\nabla g$ is $(1/L)$-cocoercive.}$$
\end{lemma}

\begin{proof}
See Theorem 2.1.5 in Section 2.1 of~\cite{Nesterov2004} 
\end{proof}

\begin{lemma}
\label{Lem:MonotoneRelationships}
Consider an operator $\Tsf: \R^n \rightarrow \R^n$ and $\beta > 0$. The following properties are equivalent
\begin{enumerate}[label=(\alph*)]
\item $\Tsf$ is $\beta$-cocoercive;
\item $\beta\Tsf$ is firmly nonexpansive;
\item $\Isf-\beta\Tsf$ is firmly nonexpansive;
\item $\beta\Tsf$ is $(1/2)$-averaged;
\item $\Isf - 2\beta \Tsf$ is nonexpansive.
\end{enumerate}
\end{lemma}

\begin{proof}
See Proposition 4 in the Supplementary Material\footnote{It can also be found in the pre-print \url{https://arxiv.org/pdf/2006.03224.pdf}.} of~\cite{Sun.etal2021}. 
\end{proof}

\begin{lemma}
\label{Lem:StrongMonotoneResidual}
Consider an operator $\Rsf = \Isf - \Dsf$ where $\Dsf: \R^n \rightarrow \R^n$.
$$\text{$\Dsf$ is Lipschitz continuous with constant $\lambda < 1$} \quad\Rightarrow\quad \text{$\Rsf$ is $(1-\lambda)$-strongly monotone.}$$
\end{lemma}

\begin{proof}
See Proposition 2 in the Supplementary Material of~\cite{Sun.etal2021}. 
\end{proof}

\begin{definition}
\label{Def:MoreauEnv}
Consider a proper, closed, and convex function $h$ and a constant $\mu > 0$. We define the \emph{proximal operator}
$$\prox_{\mu h}(\xbm) = \argmin_{\zbm \in \R^n}\left\{\frac{1}{2}\|\zbm-\xbm\|^2 + \mu h(\zbm)\right\}$$
and the \emph{Moreau envelope}
$$h_\mu(\xbm) = \min_{\zbm \in \R^n} \left\{\frac{1}{2}\|\zbm-\xbm\|^2 + \mu h(\zbm)\right\}.$$
\end{definition}

\medskip\noindent
The following two lemmas provide useful results on the Moreau envelope. The Moreau envelope is convex and smooth.
\begin{lemma}
\label{Lem:GradMorProxRes}
The function $h_\mu$ is convex and continuously differentiable with a $1$-Lipschitz gradient
$$\nabla h_\mu(\xbm) = \xbm - \prox_{\mu h}(\xbm).$$
\end{lemma}

\begin{proof}
See Proposition~8 in \cite{Sun.etal2019b}.
\end{proof}

\medskip\noindent
The Moreau envelope can also serve as a smooth approximation of a nonsmooth function.
\begin{lemma}
\label{Lem:UniformBoundMoreau}
Consider $h \in \R^n$ and its Moreau envelope $h_\mu(\xbm)$ for $\mu > 0$. Then,
$$0 \leq h(\xbm) - \frac{1}{\mu}h_\mu(\xbm) \leq \frac{\mu}{2}G_\xbm^2\quad\text{with}\quad G_\xbm^2 \defn \min_{\gbm \in \partial h(\xbm)} \|\gbm\|^2, \quad \forall \xbm \in \R^n.$$
\end{lemma}

\begin{proof}
See Proposition~9 in~\cite{Sun.etal2019b}.
\end{proof}





\bibliographystyle{IEEEtran}

\begin{thebibliography}{10}
\providecommand{\url}[1]{#1}
\csname url@samestyle\endcsname
\providecommand{\newblock}{\relax}
\providecommand{\bibinfo}[2]{#2}
\providecommand{\BIBentrySTDinterwordspacing}{\spaceskip=0pt\relax}
\providecommand{\BIBentryALTinterwordstretchfactor}{4}
\providecommand{\BIBentryALTinterwordspacing}{\spaceskip=\fontdimen2\font plus
\BIBentryALTinterwordstretchfactor\fontdimen3\font minus
  \fontdimen4\font\relax}
\providecommand{\BIBforeignlanguage}[2]{{%
\expandafter\ifx\csname l@#1\endcsname\relax
\typeout{** WARNING: IEEEtran.bst: No hyphenation pattern has been}%
\typeout{** loaded for the language `#1'. Using the pattern for}%
\typeout{** the default language instead.}%
\else
\language=\csname l@#1\endcsname
\fi
#2}}
\providecommand{\BIBdecl}{\relax}
\BIBdecl

\bibitem{Rudin.etal1992}
L.~I. Rudin, S.~Osher, and E.~Fatemi, ``Nonlinear total variation based noise
  removal algorithms,'' \emph{Physica D}, vol.~60, no. 1--4, pp. 259--268, Nov.
  1992.

\bibitem{Figueiredo.Nowak2001}
M.~A.~T. Figueiredo and R.~D. Nowak, ``Wavelet-based image estimation: An
  empirical {B}ayes approach using {J}effreys' noninformative prior,''
  \emph{IEEE Trans. Image Process.}, vol.~10, no.~9, pp. 1322--1331, Sep. 2001.

\bibitem{Elad.Aharon2006}
M.~Elad and M.~Aharon, ``Image denoising via sparse and redundant
  representations over learned dictionaries,'' \emph{IEEE Trans. Image
  Process.}, vol.~15, no.~12, pp. 3736--3745, Dec. 2006.

\bibitem{Danielyan2013}
A.~Danielyan, ``Block-based collaborative 3-{D} transform domain modeing in
  inverse imaging,'' Publication 1145, {T}ampere {U}niversity of {T}echnology,
  May 2013.

\bibitem{McCann.etal2017}
M.~T. McCann, K.~H. Jin, and M.~Unser, ``Convolutional neural networks for
  inverse problems in imaging: A review,'' \emph{IEEE Signal Process. Mag.},
  vol.~34, no.~6, pp. 85--95, 2017.

\bibitem{Lucas.etal2018}
A.~Lucas, M.~Iliadis, R.~Molina, and A.~K. Katsaggelos, ``Using deep neural
  networks for inverse problems in imaging: {B}eyond analytical methods,''
  \emph{IEEE Signal Process. Mag.}, vol.~35, no.~1, pp. 20--36, Jan. 2018.

\bibitem{Ronneberger.etal2015}
O.~Ronneberger, P.~Fischer, and T.~Brox, ``{U}-{N}et: {C}onvolutional networks
  for biomedical image segmentation,'' in \emph{Proc. Med. Image. Comput.
  Comput. Assist. Intervent.}, 2015, pp. 234--241.

\bibitem{Zhang.etal2017}
K.~Zhang, W.~Zuo, Y.~Chen, D.~Meng, and L.~Zhang, ``Beyond a {G}aussian
  denoiser: {R}esidual learning of deep {CNN} for image denoising,'' \emph{IEEE
  Trans. Image Process.}, vol.~26, no.~7, pp. 3142--3155, Jul. 2017.

\bibitem{Ongie.etal2020}
G.~Ongie, A.~Jalal, C.~A. Metzler, R.~G. Baraniuk, A.~G. Dimakis, and
  R.~Willett, ``Deep learning techniques for inverse problems in imaging,''
  \emph{IEEE J. Sel. Areas Inf. Theory}, vol.~1, no.~1, pp. 39--56, May 2020.

\bibitem{Ahmad.etal2020}
R.~{Ahmad}, C.~A. {Bouman}, G.~T. {Buzzard}, S.~{Chan}, S.~{Liu}, E.~T.
  {Reehorst}, and P.~{Schniter}, ``Plug-and-play methods for magnetic resonance
  imaging: Using denoisers for image recovery,'' \emph{IEEE Signal Processing
  Magazine}, vol.~37, no.~1, pp. 105--116, 2020.

\bibitem{Monga.etal2021}
V.~Monga, Y.~Li, and Y.~C. Eldar, ``Algorithm unrolling: {I}nterpretable,
  efficient deep learning for signal and image processing,'' \emph{IEEE Signal
  Process. Mag.}, vol.~38, no.~2, pp. 18--44, Mar. 2021.

\bibitem{Kamilov.etal2022}
U.~S. Kamilov, C.~A. Bouman, G.~T. Buzzard, and B.~Wohlberg, ``Plug-and-play
  methods for integrating physical and learned models in computational
  imaging,'' 2022, arXiv:2203.17061.

\bibitem{Venkatakrishnan.etal2013}
S.~V. Venkatakrishnan, C.~A. Bouman, and B.~Wohlberg, ``Plug-and-play priors
  for model based reconstruction,'' in \emph{Proc. IEEE Global Conf. Signal
  Process. and Inf. Process.}, Austin, TX, USA, Dec. 3-5, 2013, pp. 945--948.

\bibitem{Sreehari.etal2016}
S.~Sreehari, S.~V. Venkatakrishnan, B.~Wohlberg, G.~T. Buzzard, L.~F. Drummy,
  J.~P. Simmons, and C.~A. Bouman, ``Plug-and-play priors for bright field
  electron tomography and sparse interpolation,'' \emph{IEEE Trans. Comput.
  Imaging}, vol.~2, no.~4, pp. 408--423, Dec. 2016.

\bibitem{Gregor.LeCun2010}
K.~Gregor and Y.~LeCun, ``Learning fast approximation of sparse coding,'' in
  \emph{Proc. 27th Int. Conf. Mach. Learn.}, Haifa, Israel, Jun. 21-24, 2010,
  pp. 399--406.

\bibitem{Bora.etal2017}
A.~Bora, A.~Jalal, E.~Price, and A.~G. Dimakis, ``Compressed sensing using
  generative priors,'' in \emph{Proc. 34th Int. Conf. Machine Learning
  ({ICML})}, Sydney, Australia, Aug. 2017, pp. 537--546.

\bibitem{Gilton.etal2021}
D.~Gilton, G.~Ongie, and R.~Willett, ``Deep equilibrium architectures for
  inverse problems in imaging,'' \emph{IEEE Trans. Comput. Imag.}, vol.~7, pp.
  1123--1133, 2021.

\bibitem{Romano.etal2017}
Y.~Romano, M.~Elad, and P.~Milanfar, ``The little engine that could:
  {R}egularization by denoising ({RED}),'' \emph{SIAM J. Imaging Sci.},
  vol.~10, no.~4, pp. 1804--1844, 2017.

\bibitem{Lustig.etal2008}
M.~Lustig, D.~L. Donoho, J.~M. Santos, and J.~M. Pauly, ``Compressed sensing
  {MRI},'' \emph{IEEE Signal Process. Mag.}, vol.~25, no.~2, pp. 72--82, Mar.
  2008.

\bibitem{Bioucas-Dias.Figueiredo2007}
J.~M. Bioucas-Dias and M.~A.~T. Figueiredo, ``A new {T}w{IST}: {T}wo-step
  iterative shrinkage/thresholding algorithms for image restoration,''
  \emph{IEEE Trans. Image Process.}, vol.~16, no.~12, pp. 2992--3004, December
  2007.

\bibitem{Beck.Teboulle2009}
A.~Beck and M.~Teboulle, ``A fast iterative shrinkage-thresholding algorithm
  for linear inverse problems,'' \emph{SIAM J. Imaging Sciences}, vol.~2,
  no.~1, pp. 183--202, 2009.

\bibitem{Parikh.Boyd2014}
N.~Parikh and S.~Boyd, ``Proximal algorithms,'' \emph{Foundations and Trends in
  Optimization}, vol.~1, no.~3, pp. 123--231, 2014.

\bibitem{Figueiredo.Nowak2003}
M.~A.~T. Figueiredo and R.~D. Nowak, ``An {EM} algorithm for wavelet-based
  image restoration,'' \emph{IEEE Trans. Image Process.}, vol.~12, no.~8, pp.
  906--916, Aug. 2003.

\bibitem{Daubechies.etal2004}
I.~Daubechies, M.~Defrise, and C.~D. Mol, ``An iterative thresholding algorithm
  for linear inverse problems with a sparsity constraint,'' \emph{Commun. Pure
  Appl. Math.}, vol.~57, no.~11, pp. 1413--1457, Nov. 2004.

\bibitem{Bect.etal2004}
J.~Bect, L.~Blanc-Feraud, G.~Aubert, and A.~Chambolle, ``A $\ell_1$-unified
  variational framework for image restoration,'' in \emph{Proc. Euro. Conf.
  Comp. Vis.}, vol. 3024, New York, 2004, pp. 1--13.

\bibitem{Eckstein.Bertsekas1992}
J.~Eckstein and D.~P. Bertsekas, ``On the {D}ouglas-{R}achford splitting method
  and the proximal point algorithm for maximal monotone operators,''
  \emph{Mathematical Programming}, vol.~55, pp. 293--318, 1992.

\bibitem{Afonso.etal2010}
M.~V. Afonso, J.~M.Bioucas-Dias, and M.~A.~T. Figueiredo, ``Fast image recovery
  using variable splitting and constrained optimization,'' \emph{IEEE Trans.
  Image Process.}, vol.~19, no.~9, pp. 2345--2356, Sep. 2010.

\bibitem{Ng.etal2010}
M.~K. Ng, P.~Weiss, and X.~Yuan, ``Solving constrained total-variation image
  restoration and reconstruction problems via alternating direction methods,''
  \emph{SIAM J. Sci. Comput.}, vol.~32, no.~5, pp. 2710--2736, Aug. 2010.

\bibitem{Boyd.etal2011}
S.~Boyd, N.~Parikh, E.~Chu, B.~Peleato, and J.~Eckstein, ``Distributed
  optimization and statistical learning via the alternating direction method of
  multipliers,'' \emph{Foundations and Trends in Machine Learning}, vol.~3,
  no.~1, pp. 1--122, July 2011.

\bibitem{Gribonval.etal2012}
R.~Gribonval, G.~Chardon, and L.~Daudet, ``Blind calibration for compressed
  sensing by convex optimization,'' in \emph{Proc. IEEE Int. Conf. Acoustics,
  Speech and Signal Process.}, Kyoto, Japan, March 25-30, 2012, pp. 2713--2716.

\bibitem{Gribonval.Machart2013}
R.~Gribonval and P.~Machart, ``Reconciling ``priors'' \& ``priors'' without
  prejudice?'' in \emph{Proc. Advances in Neural Information Processing Systems
  26}, Lake Tahoe, NV, USA, December 5-10, 2013, pp. 2193--2201.

\bibitem{Reehorst.Schniter2019}
E.~T. Reehorst and P.~Schniter, ``Regularization by denoising: Clarifications
  and new interpretations,'' \emph{IEEE Trans. Comput. Imag.}, vol.~5, no.~1,
  pp. 52--67, Mar. 2019.

\bibitem{Zhang.etal2017a}
K.~Zhang, W.~Zuo, S.~Gu, and L.~Zhang, ``Learning deep {CNN} denoiser prior for
  image restoration,'' in \emph{Proc. {IEEE} Conf. Computer Vision and Pattern
  Recognition ({CVPR})}, Honolulu, USA, July 21-26, 2017, pp. 3929--3938.

\bibitem{Metzler.etal2018}
C.~Metzler, P.~Schniter, A.~Veeraraghavan, and R.~Baraniuk, ``pr{D}eep: Robust
  phase retrieval with a flexible deep network,'' in \emph{Proc. 36th Int.
  Conf. Mach. Learn.}, Stockholmsm{\"a}ssan, Stockholm Sweden, Jul. 10--15
  2018, pp. 3501--3510.

\bibitem{Dong.etal2019}
W.~{Dong}, P.~{Wang}, W.~{Yin}, G.~{Shi}, F.~{Wu}, and X.~{Lu}, ``Denoising
  prior driven deep neural network for image restoration,'' \emph{IEEE Trans.
  Pattern Anal. Mach. Intell.}, vol.~41, no.~10, pp. 2305--2318, Oct 2019.

\bibitem{Zhang.etal2019}
K.~Zhang, W.~Zuo, and L.~Zhang, ``Deep plug-and-play super-resolution for
  arbitrary blur kernels,'' in \emph{Proc. {IEEE} Conf. Computer Vision and
  Pattern Recognition ({CVPR})}, Long Beach, CA, USA, June 16-20, 2019, pp.
  1671--1681.

\bibitem{Mataev.etal2019}
G.~Mataev, P.~Milanfar, and M.~Elad, ``Deep{RED}: Deep image prior powered by
  {RED},'' in \emph{Proc. {IEEE} Int. Conf. Comput. Vis. Workshops}, Oct. 2019,
  pp. 1--10.

\bibitem{Sun.etal2019b}
Y.~{Sun}, S.~{Xu}, Y.~{Li}, L.~{Tian}, B.~{Wohlberg}, and U.~S. {Kamilov},
  ``Regularized {F}ourier ptychography using an online plug-and-play
  algorithm,'' in \emph{Proc. {IEEE} Int. Conf. Acoustics, Speech and Signal
  Process. ({ICASSP})}, Brighton, UK, May 12-17, 2019, pp. 7665--7669.

\bibitem{Liu.etal2020}
J.~{Liu}, Y.~{Sun}, C.~{Eldeniz}, W.~{Gan}, H.~{An}, and U.~S. {Kamilov},
  ``{RARE}: Image reconstruction using deep priors learned without ground
  truth,'' \emph{IEEE J. Sel. Topics Signal Process.}, vol.~14, no.~6, pp.
  1088--1099, Oct. 2020.

\bibitem{Wei.etal2020}
K.~Wei, A.~Aviles-Rivero, J.~Liang, Y.~Fu, C.-B. Sch\"onlieb, and H.~Huang,
  ``Tuning-free plug-and-play proximal algorithm for inverse imaging
  problems,'' in \emph{Proc. 37th Int. Conf. Machine Learning ({ICML})}, 2020.

\bibitem{Xie.etal2021}
M.~Xie, J.~Liu, Y.~Sun, W.~Gan, B.~Wohlberg, and U.~S. Kamilov, ``Joint
  reconstruction and calibration using regularization by denoising with
  application to computed tomography,'' in \emph{Proc. {IEEE} Int. Conf. Comp.
  Vis. Workshops ({ICCVW})}, October 2021, pp. 4028--4037.

\bibitem{Chan.etal2016}
S.~H. Chan, X.~Wang, and O.~A. Elgendy, ``Plug-and-play {ADMM} for image
  restoration: Fixed-point convergence and applications,'' \emph{IEEE Trans.
  Comp. Imag.}, vol.~3, no.~1, pp. 84--98, Mar. 2017.

\bibitem{Meinhardt.etal2017}
T.~Meinhardt, M.~Moeller, C.~Hazirbas, and D.~Cremers, ``Learning proximal
  operators: {U}sing denoising networks for regularizing inverse imaging
  problems,'' in \emph{Proc. IEEE Int. Conf. Comp. Vis. (ICCV)}, Venice, Italy,
  Oct. 22-29, 2017, pp. 1799--1808.

\bibitem{Buzzard.etal2017}
G.~T. Buzzard, S.~H. Chan, S.~Sreehari, and C.~A. Bouman, ``Plug-and-play
  unplugged: {O}ptimization free reconstruction using consensus equilibrium,''
  \emph{SIAM J. Imaging Sci.}, vol.~11, no.~3, pp. 2001--2020, Sep. 2018.

\bibitem{Ryu.etal2019}
E.~K. Ryu, J.~Liu, S.~Wang, X.~Chen, Z.~Wang, and W.~Yin, ``Plug-and-play
  methods provably converge with properly trained denoisers,'' in \emph{Proc.
  36th Int. Conf. Mach. Learn.}, vol.~97, Long Beach, CA, USA, Jun. 09--15
  2019, pp. 5546--5557.

\bibitem{Sun.etal2018a}
Y.~Sun, B.~Wohlberg, and U.~S. Kamilov, ``An online plug-and-play algorithm for
  regularized image reconstruction,'' \emph{IEEE Trans. Comput. Imag.}, vol.~5,
  no.~3, pp. 395--408, Sep. 2019.

\bibitem{Tirer.Giryes2019}
T.~Tirer and R.~Giryes, ``Image restoration by iterative denoising and backward
  projections,'' \emph{IEEE Trans. Image Process.}, vol.~28, no.~3, pp.
  1220--1234, Mar. 2019.

\bibitem{Teodoro.etal2019}
A.~M. Teodoro, J.~M. Bioucas-Dias, and M.~Figueiredo, ``A convergent image
  fusion algorithm using scene-adapted {G}aussian-mixture-based denoising,''
  \emph{IEEE Trans. Image Process.}, vol.~28, no.~1, pp. 451--463, Jan. 2019.

\bibitem{Xu.etal2020}
X.~{Xu}, Y.~{Sun}, J.~{Liu}, B.~{Wohlberg}, and U.~S. {Kamilov}, ``Provable
  convergence of plug-and-play priors with mmse denoisers,'' \emph{IEEE Signal
  Process. Lett.}, vol.~27, pp. 1280--1284, 2020.

\bibitem{Tang.Davies2020}
J.~Tang and M.~Davies, ``A fast stochastic plug-and-play {ADMM} for imaging
  inverse problems,'' 2020, arXiv:2006.11630.

\bibitem{Sun.etal2021}
Y.~Sun, Z.~Wu, B.~Wohlberg, and U.~S. Kamilov, ``Scalable plug-and-play {ADMM}
  with convergence guarantees,'' \emph{IEEE Trans. Comput. Imag.}, vol.~7, pp.
  849--863, Jul. 2021.

\bibitem{Cohen.etal2020}
R.~Cohen, M.~Elad, and P.~Milanfar, ``Regularization by denoising via
  fixed-point projection (red-pro),'' \emph{SIAM Journal on Imaging Sciences},
  vol.~14, no.~3, pp. 1374--1406, 2021.

\bibitem{Cohen.etal2021}
R.~Cohen, Y.~Blau, D.~Freedman, and E.~Rivlin, ``It has potential:
  Gradient-driven denoisers for convergent solutions to inverse problems,'' in
  \emph{Proc. Advances in Neural Information Processing Systems 34}, 2021.

\bibitem{Hurault.etal2022}
S.~Hurault, A.~Leclaire, and N.~Papadakis, ``Gradient step denoiser for
  convergent plug-and-play,'' in \emph{International Conference on Learning
  Representations ({ICLR})}, Kigali, Rwanda, May 1-5, 2022.

\bibitem{bai.etal2019}
S.~Bai, J.~Z. Kolter, and V.~Koltun, ``Deep equilibrium models,'' \emph{Proc.
  Advances in Neural Information Processing Systems 32}, 2019.

\bibitem{Gilton.etal2021a}
D.~Gilton, G.~Ongie, and R.~Willett, ``Model adaptation for inverse problems in
  imaging,'' \emph{IEEE Trans. Comput. Imag.}, vol.~7, pp. 661--674, Jul. 2021.

\bibitem{Darestani.etal2021}
M.~Z. Darestani, A.~S. Chaudhari, and R.~Heckel, ``Measuring robustness in deep
  learning based compressive sensing,'' in \emph{Proc. 38th Int. Conf. Machine
  Learning (ICML)}, July 18-24, 2021, pp. 2433--2444.

\bibitem{Darestani.etal2022}
M.~Z. Darestani, J.~Liu, and R.~Heckel, ``Test-time training can close the
  natural distribution shift performance gap in deep learning based compressed
  sensing,'' in \emph{Proc. 39th Int. Conf. Machine Learning (ICML)},
  Baltimore, MD, USA, Jul 17-23, 2022, pp. 4754--4776.

\bibitem{Jalal.etal2021}
A.~Jalal, S.~Karmalkar, A.~Dimakis, and E.~Price, ``Instance-optimal compressed
  sensing via posterior sampling,'' in \emph{Proc. 38th Int. Conf. Machine
  Learning (ICML)}, July 18-24, 2021, pp. 4709--4720.

\bibitem{Jalal.etal2021a}
A.~Jalal, M.~Arvinte, G.~Daras, E.~Price, A.~G. Dimakis, and J.~Tamir, ``Robust
  compressed sensing {MRI} with deep generative priors,'' in \emph{Proc.
  Advances in Neural Information Processing Systems 34}, Dec 6-14, 2021, pp.
  14\,938--14\,954.

\bibitem{Bertsekas2011}
D.~P. Bertsekas, ``Incremental proximal methods for large scale convex
  optimization,'' \emph{Math. Program. Ser. B}, vol. 129, pp. 163--195, 2011.

\bibitem{Schmidt.etal2011}
M.~Schmidt, N.~{Le Roux}, and F.~Bach, ``Convergence rates of inexact
  proximal-gradient methods for convex optimization,'' in \emph{Proc. Advances
  in Neural Information Processing Systems 24}, Granada, Spain, December 12-15,
  2011.

\bibitem{Devolder.etal2013}
O.~Devolder, F.~Glineur, and Y.~Nesterov, ``First-order methods of smooth
  convex optimization with inexact oracle,'' \emph{Math. Program. Ser. A}, vol.
  146, no. 1-2, pp. 37--75, 2013.

\bibitem{Gribonval2011}
R.~Gribonval, ``Should penalized least squares regression be interpreted as
  maximum a posteriori estimation?'' \emph{IEEE Trans. Signal Process.},
  vol.~59, no.~5, pp. 2405--2410, May 2011.

\bibitem{ouyang2013}
H.~Ouyang, N.~He, L.~Tran, and A.~Gray, ``Stochastic alternating direction
  method of multipliers,'' in \emph{Proceedings of the 30th International
  Conference on Machine Learning}, 2013, pp. 80--88.

\bibitem{Yu2013}
Y.~Yu, ``Better approximation and faster algorithm using the proximal
  average,'' in \emph{Neural Information Processing Systems ({NIPS})}, Lake
  Tahoe, CA, USA, December 5-10, 2013, pp. 458--466.

\bibitem{Boyd.Vandenberghe2008}
S.~Boyd and L.~Vandenberghe, ``Subgradients,'' April 2008, class notes for
  {Convex Optimization II}.
  \url{http://see.stanford.edu/materials/lsocoee364b/01-subgradients_notes.pdf}.

\bibitem{Sun2019b}
Y.~Sun, J.~Liu, and U.~S. Kamilov, ``Block coordinate regularization by
  denoising,'' in \emph{Proc. Advances in Neural Information Processing Systems
  33}, Vancouver, BC, Canada, Dec. 2019, pp. 382--392.

\bibitem{Martin.etal2001}
D.~Martin, C.~Fowlkes, D.~Tal, and J.~Malik, ``A database of human segmented
  natural images and its application to evaluating segmentation algorithms and
  measuring ecological statistics,'' in \emph{Proc. {IEEE} Int. Conf. Comp.
  Vis. ({ICCV})}, Vancouver, Canada, July 7-14, 2001, pp. 416--423.

\bibitem{zhang2018ista}
J.~{Zhang} and B.~{Ghanem}, ``{ISTA-Net}: {I}nterpretable optimization-inspired
  deep network for image compressive sensing,'' in \emph{Proc. {IEEE} Conf.
  Comput. Vision Pattern Recognit.}, 2018, pp. 1828--1837.

\bibitem{mccollough2016tu}
C.~McCollough, ``{TU-FG-207A-04}: Overview of the low dose {CT} grand
  challenge,'' \emph{Med. Phys}, vol.~43, no. 6Part35, pp. 3759--3760, 2016.

\bibitem{Liu.etal2021b}
J.~Liu, S.~Asif, B.~Wohlberg, and U.~S. Kamilov, ``Recovery analysis for
  plug-and-play priors using the restricted eigenvalue condition,'' in
  \emph{Proc. Advances in Neural Information Processing Systems 34}, December
  6-14, 2021, pp. 5921--5933.

\bibitem{Nesterov2004}
Y.~Nesterov, \emph{Introductory Lectures on Convex Optimization: A Basic
  Course}.\hskip 1em plus 0.5em minus 0.4em\relax Kluwer Academic Publishers,
  2004.

\bibitem{Anderson.etal2009}
G.~W. Anderson, A.~Guionnet, and O.~Zeitouni, \emph{An Introduction to Random
  Matrices}.\hskip 1em plus 0.5em minus 0.4em\relax {C}ambridge {U}niversity
  {P}ress, 2009.

\bibitem{Miyato.etal2018}
T.~Miyato, T.~Kataoka, M.~Koyama, and Y.~Yoshida, ``Spectral normalization for
  generative adversarial networks,'' in \emph{Int. Conf. on Learning
  Representations ({ICLR})}, Vancouver, Canada, Apr. 2018.

\bibitem{knoll2020fastmri}
{F. {Knoll} \emph{et al.}}, ``{fastMRI}: A publicly available raw k-space and
  {DICOM} dataset of knee images for accelerated {MR} image reconstruction
  using machine learning,'' \emph{Radiology: Artificial Intelligence}, vol.~2,
  no.~1, p. e190007, 2020.

\bibitem{Ong.etal2019}
F.~Ong and M.~Lustig, ``{SigPy}: a python package for high performance
  iterative reconstruction,'' in \emph{Proceedings of the ISMRM 27th Annual
  Meeting, Montreal, Quebec, Canada}, vol. 4819, 2019.

\bibitem{Bauschke.Combettes2017}
H.~H. Bauschke and P.~L. Combettes, \emph{Convex Analysis and Monotone Operator
  Theory in Hilbert Spaces}, 2nd~ed.\hskip 1em plus 0.5em minus 0.4em\relax
  Springer, 2017.

\bibitem{Ryu.Boyd2016}
E.~K. Ryu and S.~Boyd, ``A primer on monotone operator methods,'' \emph{Appl.
  Comput. Math.}, vol.~15, no.~1, pp. 3--43, 2016.

\bibitem{Rockafellar1970}
R.~T. Rockafellar, \emph{Convex Analysis}.\hskip 1em plus 0.5em minus
  0.4em\relax Princeton, NJ: Princeton Univ. Press, 1970.

\bibitem{Boyd.Vandenberghe2004}
S.~Boyd and L.~Vandenberghe, \emph{Convex Optimization}.\hskip 1em plus 0.5em
  minus 0.4em\relax Cambridge Univ. Press, 2004.

\end{thebibliography}



\end{document}